\documentclass[11pt]{article}

\usepackage{amsfonts,amsmath,amssymb,amsthm} 
\usepackage[lmargin=3cm, rmargin=3cm,bottom=2.5cm,top=2.5cm]{geometry} 

\usepackage{color}
\usepackage{soul}
\usepackage{tikz}

\newtheorem{thrm}{Theorem}[section]
\newtheorem{prprt}[thrm]{Property}
\newtheorem{fct}[thrm]{Fact}
\newtheorem{cnjctr}[thrm]{Conjecture}
\newtheorem{lmm}[thrm]{Lemma}
\newtheorem{dfntn}[thrm]{Definition}
\newtheorem{prpstn}[thrm]{Proposition}
\newtheorem{qstn}[thrm]{Question}

\theoremstyle{definition}
\newtheorem{xmpl}[thrm]{Example}

\newtheorem{clm}[thrm]{Claim}

\newcommand{\xN}{\mathbb{N}}
\newcommand{\bw}{\mathbf{w}}

\newcommand{\BPLP}{\mbox{BPLP}}
\newcommand{\BPLF}{\mbox{BPLF}}
\newcommand{\nextt}{\mbox{next}}
\newcommand{\pref}{\text{Pref}}
\newcommand{\fact}{\text{Fact}}
\newcommand{\pal}[1]{|#1|_{pal}}
\newcommand{\lgpal}[1]{|#1|_{LGPal}}
\newcommand{\rgpal}[1]{|#1|_{RGPal}}
\newcommand{\SWHIMPP}{{\cal P}(A^\omega)}
\newcommand{\maxlgpalpref}[1]{{\rm MaxLGPalPref}(#1)} 
\newcommand{\maxrgpalpref}[1]{{\rm MaxRGPalPref}(#1)} 

\begin{document}
\title {Greedy palindromic lengths\thanks{Parts of this paper were presented at the conference Journ\'ees Montoises 2014 }
~\thanks{This work was partially done while the first author was visiting the LIRMM. 
Many thanks to ANR-09-BLAN-0164 EMC for funding this visit.}}

\author{Michelangelo Bucci\thanks{michelangelo.bucci@gmail.com}, 
Gwena\"el Richomme\thanks{gwenael.richomme@lirmm.fr, LIRMM (CNRS, Univ. Montpellier 2) - UMR 5506 - CC 477, 161 rue Ada, 34095, Montpellier Cedex 5 - France}~\thanks{Univ. Paul-Val\'ery Montpellier 3, Dpt MIAp, Route de Mende, 34199 Montpellier Cedex 5, France}}
\maketitle


\begin{abstract}
In [A. Frid, S. Puzynina, L.Q. Zamboni, 
\textit{On palindromic factorization of words}, 
Adv. in Appl. Math. 50 (2013), 737-748], 
it was conjectured that any infinite word whose palindromic lengths of factors are bounded is ultimately periodic. 
We introduce variants of this conjecture and 
prove this conjecture in particular cases. 
Especially we introduce left and right greedy palindromic lengths. 
These lengths are always greater than or equals to the initial palindromic length.
When the greedy left (or right) palindromic lengths of prefixes of a word 
are bounded then this word is ultimately periodic.
\end{abstract}


\maketitle

\section{Introduction}

A fundamental question in Combinatorics on Words is how words can be decomposed on smallest words.
For instance, readers can think to some topics presented in the first Lothaire's book \cite{Lothaire1983book} like Lyndon words, critical factorization theorem, equations on words, or to the theory of codes \cite{BerstelPerrin1985book,BerstelPerrinReutenauer2010book}, or to many related works published since these surveys. As another example, let us mention that 
in the area of Text Algorithms some factorizations like Crochemore or Lempel-Ziv factorizations play an important role \cite{Crochemore1983CRAS,LempelZiv1976IEEETIT}. These factorizations have been extended to infinite words and, for the Fibonacci word, some links have been discovered with the Wen and Wen's decomposition in singular words \cite{BerstelSavelli2006MFCS}  (see also  \cite{GhareghaniMohammad-NooriSharifani2011TCS} for a generalization to Sturmian words and see \cite{WenWen1994EJC,LeveSeebold2003BBMS} for more on singular words).

In their article~\cite{FridPuzyninaZamboni2013AAM}, 
A.E.~Frid, S.~Puzynina and L.Q.~Zamboni defined the \textit{palindromic length} of a finite word $w$ 
as the least number of (nonempty) palindromes needed to decompose $w$. 
More precisely the palindromic length of $w$ is the least number $k$ such that $w = \pi_1 \cdots \pi_k$ 
with $\pi_1$, \ldots, $\pi_k$ palindromes. 
As in the article \cite{FridPuzyninaZamboni2013AAM} 
we let $\pal{w}$ denote the palindromic length of $w$. For instance $\pal{abaab} = 2$.
The palindromic length of a finite word can also be defined inductively by $\pal{\varepsilon} = 0$ (where $\varepsilon$ denotes the empty word), and, for any nonempty word $u$, $\pal{u} = \min\{ \pal{p} + 1 \mid u = p \pi$ and $\pi$ is a palindromic nonempty suffix of $u\}$. This recursive definition is the starting point of G.~Fici, T.~Gagie, J.~K\"arkk\"ainen and D.~Kempa \cite{Fici_Gagie_Karkkainen_Kempa2014JDA} for providing an algorithm determining, in time $O(nlog(n))$ and space $O(n)$, the palindromic length of a word of length $n$. In the article~\cite{Shur_Rubinchik2015Iwoca}, A.~Shur and M.~Rubinchik provide another algorithm with the same complexities using a structure storing palindromes. 
They also conjecture that the palindromic length can be computed in time $O(n)$.

Let us recall that, for any nonempty word $v$, 
the infinite periodic word with \textit{period} $v$ ($v \neq \varepsilon$) is denoted by $v^\omega$. Hence any ultimately periodic word is 
in the form $uv^\omega$ for some words $u$ and $v$ ($v \neq \varepsilon)$.
A.E.~Frid, S.~Puzynina and L.Q.~Zamboni conjectured:

\begin{cnjctr}{\rm \cite{FridPuzyninaZamboni2013AAM}}
\label{conj1}
If the palindromic lengths of factors of an infinite word $\bw$ are bounded 
(we will say that $\bw$ has bounded palindromic lengths of factors) then $\bw$ is ultimately periodic.
\end{cnjctr}

In Section~\ref{sec:variants}, we consider several variants of this conjecture and show the equivalence between some of them. In particular, we show that the previous conjecture can be restricted to words having infinitely many palindromic prefixes and that, with this restriction, one should expect words to be periodic instead of being ultimately periodic. 
This is a consequence of an intermediate result (Lemma~\ref{L:ex2.1(2)})
that states that, if a word has bounded palindromic lengths of prefixes,
then it has a suffix that has infinitely many palindromic prefixes.
Section~\ref{sec:about_pal_A_omega} shows that this is not just a consequence
of the fact that words considered in this result have infinitely many palindromic factors. Indeed we provide an example showing that a word containing infinitely many palindromes may not have a suffix that begins with infinitely many palindromes, even if the considered word is uniformly recurrent.
In Section~\ref{sec:around}, we give a characterization of words having infinitely many palindromic prefixes and whose palindromic lengths of prefixes or factors are bounded by 2. 
As they are all periodic, this proves the conjectures in a special case.

Let us mention two reasons for the difficulty to prove Conjecture~\ref{conj1} in the general case. First when the palindromic lengths of prefixes of an infinite word are not bounded, the function that associates with each integer $k$ the length of the smallest prefix having palindromic length $k$ can grow very slowly.
For instance, let us consider the Fibonacci infinite word. As it is the fixed point of the morphism $\varphi$ defined by $\varphi(a) = ab$ and $\varphi(b) = a$, by \cite{FridPuzyninaZamboni2013AAM}, the palindromic lengths of prefixes are not bounded. Actually, if $m(k)$ denotes the length of the least nonempty prefix of the Fibonacci word with palindromic length $k$, one can verify that $m(1) = 1$, $m(2) = 2$, $m(3) = 9$, $m(4) = 62$, $m(5) = 297$, $m(6) = 1154$, $m(7) = 5473$.
The second source of difficulties lies in the facts that a word may have several minimal palindromic factorizations and that the minimal palindromic factorizations of a word and of its longest proper prefixes are not related. For instance both words $aabaab$ and $aabaaba$ have palindromic length 2. The first word has two corresponding minimal palindromic factorizations: $aabaa.b$ and $aa.baab$. It is the longest proper prefix of $aabaaba$ who admits only one minimal palindromic decomposition: $a.abaaba$. To cope with the previous difficulty, in Section~\ref{sec:greedypal}, we introduce greedy palindromic lengths. For instance, the left greedy palindromic length of a word is the number of palindromes in the palindromic decomposition obtained considering iteratively the longest palindromic prefix as first element. We show that if the left (or the right) greedy palindromic lengths of prefixes of an infinite word $\bw$ having infinitely many palindromic prefixes are bounded then $\bw$ is periodic. As it also implies that palindromic lengths of factors of $\bw$ are bounded, this proves the Frid, Puzynina and Zamboni's conjecture in a special case.

We assume readers are familiar with combinatorics on words. We report them to classical surveys, as  \cite{Lothaire1983book,Lothaire2002book} for instance, for basic definitions as those already used (word, length, palindrome, prefix, suffix, \ldots). From now on, $A$ denotes an alphabet. If $u = u_1 \cdots u_k$ is a word (with each $u_i$ a letter), we let $\tilde{u}$ denote the mirror image of $u$ that is the word $u_k \cdots u_1$. We also recall that a finite word $w$ is \textit{primitive} if $w$ is not a power of a smaller word. It is well-known that such a primitive word $w$ is not an internal factor of $w$, that is $ww = xwy$ implies $x$ or $y$ is the empty word.

\section{\label{sec:variants}Variants}

In this section, we introduce several conjectures related to Conjecture~\ref{conj1} and discuss links between them.

Let us first recall some results on palindromic lengths
obtained by A.E.~Frid, S.~Puzynina and L.Q.~Zamboni \cite{FridPuzyninaZamboni2013AAM}.
By a counting argument, they proved that any infinite word with bounded palindromic lengths of prefixes (and not only factors) contains $k$-powers for arbitrary integers $k$. Moreover, if the word is not ultimately periodic, each position of the word must be covered by infinitely many runs. Thus the next conjecture could have been formulated in the article~\cite{FridPuzyninaZamboni2013AAM}.

\begin{cnjctr}
\label{conj:prefixes}
If the palindromic lengths of prefixes of an infinite word $\bw$ are bounded then $\bw$ is ultimately periodic.
\end{cnjctr}

Obviously if Conjecture~\ref{conj:prefixes} was verified then Conjecture~\ref{conj1} would also be verified. Nevertheless we do not know how to prove the converse (except by proving the conjectures). Hence Conjecture~\ref{conj:prefixes} seems more difficult than Conjecture~\ref{conj1}.

We let $\BPLF$ (resp. $\BPLP$) denote the set of all infinite words whose palindromic lengths of factors (resp. prefixes) are bounded. Thus Conjecture~\ref{conj1} (resp. \ref{conj:prefixes}) could be rewritten: Any word in $\BPLF$ (resp. $\BPLP$) is ultimately periodic.

Observe that any element of the sets $\BPLF$  and $\BPLP$ contains infinitely many palindromes.
We let $\SWHIMPP$ denote the set of all infinite words over the alphabet $A$ having infinitely many palindromic prefixes. 

Next result justify our two new variants of Conjecture~\ref{conj1}.
\begin{lmm}[\mbox{\cite[Lem. 5.6]{Fischler2006JCT}}]
\label{Fischler}
Any ultimately periodic word belonging to $\SWHIMPP$ is periodic.
\end{lmm}

\begin{cnjctr}
\label{conj2}
Any infinite word in $\BPLF \cap \SWHIMPP$ is periodic.
\end{cnjctr}

\begin{cnjctr}
\label{conj:prefixes_infini_prefixes_palindromes}
Any infinite word in $\BPLP \cap \SWHIMPP$ is periodic.
\end{cnjctr}

As for Conjectures~\ref{conj1} and \ref{conj:prefixes}, if Conjecture~\ref{conj:prefixes_infini_prefixes_palindromes} is verified then Conjecture~\ref{conj2} is also verified, and, we do not know any way to prove directly the converse. Conjecture~\ref{conj:prefixes_infini_prefixes_palindromes} seems more difficult than Conjecture~\ref{conj2}.

Conjectures~\ref{conj2} and \ref{conj:prefixes_infini_prefixes_palindromes} may be easier to tackle as we have more information on the considered word than in Conjectures~\ref{conj1} and \ref{conj:prefixes}. 

\begin{prpstn}
\label{P:equivalence1}
Conjecture~\ref{conj1} and \ref{conj2} are equivalent.
\end{prpstn}

\begin{prpstn}
\label{P:equivalence2}
Conjecture~\ref{conj:prefixes} and \ref{conj:prefixes_infini_prefixes_palindromes} are equivalent.
\end{prpstn}

Proposition~\ref{P:equivalence1} is a corollary of the following result.

\begin{lmm}
\label{L:ex2.1(2)}
If the palindromic lengths of prefixes (or more generally of factors) of an infinite word $\bw$ are bounded, then
there exists a suffix $\bw'$ of $\bw$ having infinitely many palindromic prefixes.
\end{lmm}

\begin{proof}
Let $I_0 = \{0\}$ and for $k \geq 1$, let\\\centerline{
$I_k = \{ i \mid \exists j \in I_{k-1}, \bw[j+1..i] \mbox{ is a palindrome}\}$.}
Note that for $k \geq 1$, $I_k$ is the set of all lengths of prefixes of $\bw$ that can be decomposed into $k$ palindromes.
If all suffixes of $\bw$ have only a finite number of palindromic prefixes, then it can be checked quite directly by induction that, for all $k \geq 0$, $I_k$ is finite. This contradicts the fact that $\bw$ belongs to BPLP. Thus there exists a smallest integer $k \geq 1$ such that $I_k$ is infinite and $I_{k-1}$ is finite. So there exists $j \in I_{k-1}$, such that $\bw[j+1..\infty]$ has infinitely many palindromic prefixes.
\end{proof}

\begin{proof}[Proof of Proposition~\ref{P:equivalence1}]
Let us first assume that Conjecture~\ref{conj1} holds and let $\bw$ be an element in $\BPLF \cap \SWHIMPP$: $\bw$ is ultimately periodic. 
By Lemma~\ref{Fischler}, $\bw$ is periodic:
Conjecture~\ref{conj2} holds.
Observe now that any suffix of an infinite word in $\BPLF$ is also in $\BPLF$.
Thus by Lemma~\ref{L:ex2.1(2)}, if $\bw \in \BPLF$, it has a suffix $\bw'$ in $\BPLF \cap \SWHIMPP$.
Hence Conjecture~\ref{conj2} implies Conjecture~\ref{conj1}.
\end{proof}

Proposition~\ref{P:equivalence2} is also a consequence of Lemma~\ref{L:ex2.1(2)}. 
Nevertheless its proof also needs next lemma that states that any suffix of an element of $\BPLP$ also belongs to $\BPLP$.
For any integer $k \geq 1$, we let $\BPLF(k)$ (resp. $\BPLP(k)$) denote
the set of all infinite words such that $\pal{u} \leq k$ for all their factors $u$ (resp. all their prefixes $u$).

\begin{lmm}
\label{L:ex2.1(1)}
Let $k \geq 1$ be an integer, let $\bw$ be an infinite word and let $a$ be a letter.
If $a \bw \in \BPLP(k)$, then $\bw \in \BPLP(k+1)$
\end{lmm}

\begin{proof}
Let $p$ be a finite word. 
If $ap = \pi_1 \pi_2 \cdots \pi_k$ with $\pi_1$,  \ldots, $\pi_k$ palindromes, then $p = \pi_2 \ldots \pi_k$ if $\pi_1 = a$, and $p = \pi'_1 a \pi_2 \ldots \pi_k$ if $\pi_1 = a \pi_1' a$. Thus if $a\bw$ belongs to $\BPLP(k)$, $\bw$ belongs to $\BPLP(k+1)$. In particular if $a\bw \in BPLP$ then $\bw \in BPLP$. 
\end{proof}

\begin{proof}[Proof of Proposition~\ref{P:equivalence2}]
First if Conjecture~\ref{conj:prefixes} is verified, then Lemma~\ref{Fischler} implies that Conjecture~\ref{conj:prefixes_infini_prefixes_palindromes} is also verified.

Now assume that Conjecture~\ref{conj:prefixes_infini_prefixes_palindromes} is verified.
Let $\bw$ be in $\BPLP$.
By Lemma~\ref{L:ex2.1(2)}, $\bw$ has a suffix $\bw'$ in $\SWHIMPP$. 
Lemma~\ref{L:ex2.1(1)} implies that $\bw'$ belongs to $\BPLP$. Thus as Conjecture~\ref{conj:prefixes_infini_prefixes_palindromes} is verified, $\bw'$ is periodic: $\bw$ is ultimately periodic. Conjecture~\ref{conj:prefixes} is verified.
\end{proof}

It follows immediately from the definition of $\BPLF$, that for $k \geq 1$ and 
$\bw$ an infinite word,  if $a\bw \in \BPLF(k)$, then $\bw \in \BPLF(k)$.
Thus one can ask whether Lemma~\ref{L:ex2.1(1)} can be improved. 
But, in general, 
$a\bw \in \BPLP(k)$ does not imply $\bw \in \BPLP(k)$.
For instance $a(abba)^\omega \in \BPLP(2)$ while $(abba)^\omega \in \BPLP(3)\setminus \BPLP(2)$.

\section{\label{sec:about_pal_A_omega}About the hypothesis of Conjectures~\ref{conj2} and \ref{conj:prefixes_infini_prefixes_palindromes}}

As any word in $\BPLF$ contains infinitely many palindromes, one could ask whether Lemma~\ref{L:ex2.1(2)} is a special case of a more general result, that is, 
whether any word containing infinitely many palindromes 
has a suffix in $\SWHIMPP$.
The word $\prod_{i=1}^\infty ab^i = ababbabbbabbbba\cdots$ 
shows this is not the case. 
This word is not \textit{recurrent}, that is, it has some factors (at least all its prefixes) that does not occur infinitely often.
The following result provides also a negative answer for \textit{uniformly recurrent} words, that is, words whose all factors occur infinitely often with bounded gaps.

\begin{dfntn}
Let $(u_n)_{\geq 0}$ be the sequence of words defined by $u_0 = aa$ and $u_{n+1} = u_n bbab u_n \widetilde{u_n}$, and 
let $U = \lim_{n \to \infty} u_n$.
\end{dfntn}

\begin{lmm}
\label{LnoSuffixWithInfinitelyManyPalindromicPrefixes}
The word $U$ defined above is a binary uniformly recurrent word containing infinitely many palindromes but such that none of its suffixes begins with infinitely many palindromes.
\end{lmm}

\begin{proof}
One can observe that $|u_n| = 4.3^n -2$ and 
$U = aa \prod_{n \geq 0} bbab u_n \widetilde{u_n}$. 
Thus the sequence $(u_n \widetilde{u_n})_{n \geq 0}$ is an infinite sequence of pairwise different palindromic factors of $U$.
As $(u_n)_{n \geq 0}$ is a sequence of prefixes of $U$, it follows that the set of factors of $U$ is closed under reversal (any factor $v$ of $U$ occurs in a prefix $u_n$ and so $\tilde{v}$ occurs in $\widetilde{u_n}$ itself a factor of $U$).

Let us prove that $U$ is uniformly recurrent.
Let $n$ be an integer. There exists an integer $k$ depending on $n$ such that all factors of length $n$ occur in the prefix $u_k$ of $U$. As the set of factors of $U$ is closed under reversal, all factors of length $n$ also occur in $\widetilde{u_k}$. From the definition of the sequence $(u_i)_{i \geq 0}$, one can inductively prove that each word $u_N$ with $N \geq k$ can be decomposed over the set $\{ u_k, u_k bbab, u_k babb, \widetilde{u_k}, \widetilde{u_k} bbab, \widetilde{u_k} babb\}$. Consequently $U$ can also be decomposed over this set and the distance between two occurrences of a same factor of length $n$ is at most $2|u_k|+4$. The word $U$ is uniformly recurrent.

To prove that any suffix of $U$ has a finite number of palindromic prefixes, let us first observe that

\begin{fct}
\label{Fdiff}
For each prefix $p$ of $u_n\widetilde{u_n}$ ($n \geq 0$), $|p|_{bbab}-|p|_{babb} \geq 0$. 
\end{fct}

\begin{proof}
We prove the fact by induction on $n$. As $u_0 = aa$, the fact is immediate for $n = 0$.
Let $p$ be a prefix of $u_{n+1}$. One of the following cases holds:
\begin{description}
\item[Case 1:] $p$ is a prefix of $u_n$, 
\item[Case 2:] $p = u_nq$ with $q \in \{b, bb, bba\}$, 
\item[Case 3:] $p = u_n bbab p'$ with $p'$ a prefix of $u_n$ or
\item[Case 4:] $p = u_n bbab p' q \tilde{q}$ with $p'$ a prefix of $u_n$ and $q$ such that $u_n = p'q$.
\end{description}
For any word $w$, let $\Delta(w)$ denote $|w|_{bbab}-|w|_{babb}$. 
Observe that $aa$ is both a prefix and a suffix of $u_n$.
Thus $\Delta(u_nq) = \Delta(u_n)$ for any $q$ in $\{b, bb, bba\}$. 
Consequently Fact~\ref{Fdiff} is a consequence of the inductive hypothesis in Cases 1 and 2. 
In Case~3, one has $\Delta(p) = \Delta(u_n) + \Delta(p') + 1$ and once again the result holds as a direct consequence of the inductive hypothesis. 
In Case~4, one has $\Delta(q \tilde{q}) = 0$. Thus $\Delta(p) = \Delta(u_n) + \Delta(p') + \delta$ with $\delta \in \{0, 1, 2\}$ ($0$ if $babb$ overlaps the end of $p'$ and the beginning of $q$, $2$ if $bbab$ overlaps the end of $p'$ and the beginning of $q$, $1$ otherwise). Once again, the result holds by induction.
\end{proof}

Let $S$ be a suffix of $U = aa \prod_{n \geq 0} bbab u_n \widetilde{u_n}$. 
There exists an integer $N$ and a word $s$ such that $S = s \prod_{n \geq N} bbab u_n \widetilde{u_n}$. 
Let $k = \max(0, |s|_{babb}-|s|_{bbab}) = max(0, \Delta(s))$. 
A consequence of Fact~\ref{Fdiff} is that, if $p = \left[\prod_{n = n_1}^{n_2} bbab u_n \widetilde{u_n}\right] p'$ with $p'$ a prefix of $bbabu_{n+1}\widetilde{u_{n+1}}$, 
then $\Delta(p) \geq n_2-n_1+1$.
Thus for any prefix $\pi$ of $S$ longer than 
$s\prod_{n = N}^{N+k} bbab u_n \widetilde{u_n}$, $|\pi|_{bbab}-|\pi|_{babb} \geq 1$ and so $p$ cannot be a palindrome: $S$ has a finite number of palindromic prefixes. This ends the proof of Lemma~\ref{LnoSuffixWithInfinitelyManyPalindromicPrefixes}.
\end{proof}

\section{\label{sec:around}Another question and a study of bound 2}

We have already mentioned that, as $\BPLF \subseteq \BPLP$, 
Conjecture~\ref{conj1} implies Conjecture~\ref{conj2}, 
and Conjecture~\ref{conj:prefixes} implies Conjecture~\ref{conj:prefixes_infini_prefixes_palindromes}. 
Conversely we do not know whether $\BPLP \subseteq \BPLF$. Such a result would be useful to show the equivalence between all previous conjectures.
In this context, as for  all $k \geq 1$, 
$\BPLF(k) \subseteq \BPLP(k)$, the following question is of interest.
\begin{qstn}
Given an integer $k \geq 1$, what is the smallest integer $K$ (if it exists) such that 
$\BPLP(k) \subseteq \BPLF(K)$?
\end{qstn}

As $\BPLP(1) = \BPLF(1) = \{ a^\omega \mid a$ letter $\}$, this question holds for $k \geq 2$.
Proposition~\ref{P:bound2} answers it for $k = 2$ in the restricted context of words of $\SWHIMPP$ (a restriction authorized after Section~\ref{sec:variants} and that reduces the number of cases to study).
Note that to prove Proposition~\ref{P:bound2}, we characterize elements of $\BPLF(2) \cap \SWHIMPP$ and 
$\BPLP(2) \cap \SWHIMPP$.
Observing that all these words are periodic confirms the conjectures.

\begin{prpstn}
\label{P:bound2}
$\BPLF(2) \cap \SWHIMPP \subsetneq \BPLP(2) \cap \SWHIMPP \subsetneq \BPLF(3) \cap \SWHIMPP$.
\end{prpstn}

The proof of this proposition is a direct consequence of Lemma~\ref{L:BPLF2} and Lemma~\ref{L:approche2bis} below.
Indeed:
\begin{itemize}
\item from the definitions, we have $\BPLF(2) \subseteq \BPLP(2)$; 

\item from Lemmas~\ref{L:BPLF2} and \ref{L:approche2bis}, we observe that $(ababa)^\omega \in \BPLP(2) \cap \SWHIMPP\setminus \BPLF(2)$;

\item all words occurring in Lemma~\ref{L:approche2bis} belong to $\BPLF(3)$ (the exhaustive verification of palindromic lengths of factors is left to readers); 

\item the word $(abba)^\omega$ belong to 
$\BPLF(3) \cap \SWHIMPP\setminus \BPLP(2)$ (the exhaustive verification of palindromic lengths of factors is left to readers).
\end{itemize}

\begin{lmm}
\label{L:BPLF2}
The set $(\BPLF(2)\setminus \BPLF(1))\cap \SWHIMPP$ is the set of all words 
in the form $a^i(ba^j)^\omega$ with $a$, $b$ two different letters, 
and $i$, $j$ two integers such that $0 \leq i \leq j$ and $j \neq 0$.
\end{lmm}

\begin{proof}
Let $\bw$ be a word in $(\BPLF(2)\setminus \BPLF(1))\cap \SWHIMPP$.
As $\bw$ begins with infinitely many palindromes, each factor occurs infinitely often.
Condition $\bw \in \BPLF(2)$ implies that $\bw$ is a binary word. Indeed, otherwise it would
have a factor in the form $ab^ic$ for an integer $i \geq 1$ and some different letters $a$, $b$ and $c$.

If $\bw$ has no factor in the form $aa$ with $a$ a letter, 
as $\bw$ is written exactly on two letters (otherwise it would belong to $\BPLF(1)$),
$\bw = a^i(ba)^\omega$ with $i \leq 1$ and $a$, $b$ two different letters. The result holds.

From now on assume that $\bw$ has a factor in the form $aa$ with $a$ a letter. 
Let $b$ be the second letter occurring in $\bw$. 
The word $\bw$ has a factor in the form $ba^jb$ with $j \geq 2$ (we have initially pointed out that all factors, hence factors $aa$ and $b$, occur infinitely often in $\bw$).
Let $\bw'$ be a suffix of $\bw$ beginning with $ba^jb$.
Let $p_n$ be the smallest prefix of $\bw'$ that contains exactly $n$ occurrences of the letter $b$: $p_0 = \varepsilon$, $p_1 = b$, $p_2 = ba^jb$. 

Observe that:
\begin{enumerate}
\item for all $k \geq 2$, $\pal{ba^jb^ka} = 3$;
\item for all $k \geq 1$, $\pal{ba^jba^kb} = 3$ if $j \neq k$.
\end{enumerate}  
Using these observations, it follows by induction that $p_n = (ba^j)^{n-1}b$ for all $n \geq 1$.
Hence $\bw' = (ba^j)^\omega$.

As $\bw$ begins with infinitely many palindromes, 
prefixes of $\bw$ are mirror images of some factors of $\bw'$.
Thus $\bw = a^i(ba^j)^\omega$ for some integer $i$ such that $0 \leq i \leq j$.
\end{proof}

Next result provides a characterization of words in $\BPLP(2) \cap \SWHIMPP$.

\begin{lmm}
\label{L:approche2bis}
An infinite word $\bw$ beginning with the letter $a$ and having infinitely many palindromic prefixes 
is in $\BPLP(2)$ if and only if it is in one of the following forms ($b$ is a letter different from $a$):
\begin{enumerate}
\itemsep0cm
\item $\bw = a^\omega$;
\item $\bw = (a^iba^j)^\omega$ for some integers $i \geq 1$, $j \geq 1$;
\item $\bw = (a^ib^j)^\omega$ for some integers $i \geq 1$, $j \geq 1$;
\item $\bw = ((ab)^ia)^\omega$ for some integer $i \geq 2$.
\end{enumerate}
\end{lmm}

A first step for the proof of the lemma is next result.
\begin{lmm}
\label{L:size_alphabet}
Any infinite word in $\BPLP(2) \cap \SWHIMPP$ contains at most two different letters.
\end{lmm}

\begin{proof}
Assume that $\bw$ contains at least three letters. Then it has a prefix in the form $pc$ with $p$ containing exactly two different letters $a$ and $b$, and with $c$ a letter different from $a$ and $b$. Inequality $\pal{pc}\leq 2$ implies $\pal{p} = 1$. 

The word $\bw$ has a prefix in the form $pc^id$ with $d$ a letter different from $c$, 
and $i \geq 1$ an integer (As $\bw \in \SWHIMPP$, $\bw \neq pc^\omega$). Conditions ``$\pal{pc^id} \leq 2$" and ``$c$ does not occur in $pd$ implies $p = \varepsilon$, $p = d$ or $p = \pi d$ for a palindrome $\pi$. In the latter case, the fact that both $\pi$ and $p$ are palindromes implies that $\pi$ is a power of $d$. In all cases, we have a contradiction with the fact that $p$ contains two different letters.
\end{proof}

In order to prove Lemma~\ref{L:approche2bis}, we consider the possible palindromic prefixes of $\bw$. 
We let $\nextt(u)$ denote the set of all palindromes $\pi$ over $\{a, b\}$ such that:
\begin{itemize}
\item $u$ is a proper prefix of $\pi$ (A word $u$ is a \textit{proper prefix} of a word $v$ if there exists a non empty word $p$ such that $v = pu$),
\item all proper palindromic prefixes of $\pi$ are prefixes of $u$ (possibly $u$ itself),
and 
\item $\pal{p} \leq 2$ for all the prefixes $p$ of $\pi$.
\end{itemize}
For instance $\nextt(a^i) = \{a^{i+1}\} \cup \nextt(a^i b)$.

\begin{lmm}
\label{L:approche2}For any integer $i \geq 1$,
\begin{enumerate}
\itemsep0cm
\item \label{item1} $\nextt(a^i b) = \{a^i b^j a^i \mid j \geq 1\} \cup \{a^i(ba^j)^kba^i \mid k \geq 1, 1 \leq j < i \}$;
\item \label{item2} $\nextt(a^i(ba^j)^kba^i) = \emptyset$ when $1 \leq j < i$, $k \geq 1$;
\item \label{item3} $\nextt(a^i (b^j a^i)^ka) = \emptyset$ when $j \geq 2$ and $k \geq 1$;

\item \label{item4} $\nextt(a^i (b^j a^i)^kb) = \{ a^i (b^j a^i)^{k+1} \}$ when $j \geq 1$ and $k \geq 1$;

\item \label{item5} $\nextt(a^i (b a^i)^ka) = \emptyset$ when $i \geq 2$ and $k \geq 2$;

\item \label{item6} $\nextt(a^i b a^{i+1}) = \{ a^iba^{i+j}ba^i \mid j \geq 1\}$.

\item \label{item7} $\nextt(a^i (b a^{i+j})^kba^i) = \{ a^i (b a^{i+j})^{k+1}ba^i\}$ when $j \geq 1$, $k \geq 1$;

\item \label{item8} $\nextt(a(ba)^ka) = \{ (a(ba)^k)^2 \}$ when $k \geq 2$.

\item \label{item9} $\nextt((a(ba)^k)^\alpha) = \{ (a(ba)^k)^{\alpha+1} \}$ when $k \geq 2$ and $\alpha \geq 2$
\end{enumerate}
\end{lmm}

\begin{proof}
\begin{enumerate}

\item Let $E_1 = \{a^i b^j a^i \mid j \geq 1\} \cup \{a^i(ba^j)^kba^i \mid k \geq 1, 1 \leq j < i \}$.
We let readers verify that $E_1 \subseteq \nextt(a^i b)$.
Assume by contradiction that there exists a word $u$ in $\nextt(a^i b)\setminus E_1$.
Let $\pref(E_1)$ denote the set of all prefixes of words in $E_1$.
The word $u$ must have a prefix $p$ in $(\pref(E_1)\setminus E_1)\{a, b\} \setminus \pref(E_1)$.
In other words $p = \pi a$ or $p = \pi b$ with $p \not\in \pref(E_1)$,
$\pi \in \pref(E_1)$ and $\pi \not\in E_1$.
Observe that 
$\pref(E_1) = 
   \{a^\ell \mid \ell \leq i\} 
   \cup \{a^i b^j \mid j \geq 0\}
   \cup \{a^i b^ja^k \mid j \geq 1, 0 \leq k \leq i\}
   \cup \{a^i(ba^j)^{k} b a^\ell \mid k \geq 0, 0 \leq \ell \leq i, 1 \leq j < i\}$.
As $a^ib$ is a prefix of $u$ and $\pi \not\in E_1$, we get $p = a^ib^ja^k b$
with $j \geq 2$ and $1 \leq k < i$, or,
$p = a^i(ba^j)^kba^\ell b$ 
with $1 \leq j < i$, $k \geq 1$, $0 \leq \ell < i $ and $\ell \neq j$. In these two cases, $\pal{p} = 3$. 
Hence $\nextt(a^i b) = E_1$.

\item Observe that any palindrome having $a^i(ba^j)^kba^i$ as a proper prefix must have a prefix in one of the two following forms: $p = a^i(ba^j)^kba^{i+\ell}b$ with $\ell \geq 1$; $p = a^i(ba^j)^kba^{i}b^ma$ with $m \geq 1$.
As $1 \leq j < i$ and $k \geq 1$, 
in both cases $\pal{p} = 3$. Thus $\nextt(a^i(ba^j)^kba^i) = \emptyset$.

\item 
Any element of $\nextt(a^i (b^j a^i)^ka)$ must have a prefix $p = a^i(b^ja^i)^ka^\ell b$ with $\ell \geq 1$.
Once again, as $j \geq 2$ and $k \geq 1$, $\pal{p} = 3$.
Hence $\nextt(a^i(b^ja^i)^ka) = \emptyset$.

\item 
We let readers verify that $a^i(b^ja^i)^{k+1} \in \nextt(a^i(b^ja^i)^kb)$.
Assume there exists a word $u$ in $\nextt(a^i(b^ja^i)^kb)\setminus \{a^i(b^j a^i)^{k+1}\}$.
It has a prefix $p$ in the form $a^i(b^ja^i)^kb^\ell a$ with $\ell \neq j$ and $\ell \geq 1$ or
 in the form $a^i(b^ja^i)^kb^j a^\ell b$ with $1 \leq \ell < i$.
 Thus, as $j \geq 1$ and $k \geq 1$, we have $\pal{p} = 3$ except if $j = 1$ for the second form. But then $u$ has $pa$ or $pb$ as a prefix with $p = a^i(ba^i)^kba^\ell b$. As $\pal{pa} = \pal{pb} = 3$, we get a contradiction.
 Hence  $\nextt(a^i(b^ja^i)^kb) = \{a^i(b^j a^i)^{k+1}\}$.

\item Let $i \geq 2$, $k \geq 2$ and let $u$ be an element of $\nextt(a^i (b a^i)^ka)$.
For some $\ell \geq 1$ and $m \geq 1$, the word $p_1 := a^i (b a^i)^{k-1}ba^{i+\ell}b^ma$ is a prefix of $u$. 
Observe that $\pal{p_1} \leq 2$ only if $\ell = 1$.
Then $u$ has one of the following two words as a prefix: $p_2 := a^i (b a^i)^{k-1}ba^iab^maa$ or $p_3 := a^i (b a^i)^{k-1}ba^iab^mab$. 
But, as $i \geq 2$ and $k \geq 2$, $\pal{p_2} = 3 = \pal{p_3}$: a contradiction with $u \in \BPLP(2)$.
Thus $\nextt(a^i (b a^i)^ka) = \emptyset$.

\item We let readers verify that $\{ a^iba^{i+j}ba^i \mid j \geq 1\} \subseteq \nextt(a^i b a^{i+1})$.
Now assume that $u$ is an element of $\nextt(a^i b a^{i+1})$. 
It has a prefix $p_1 := a^iba^{i+j} b^\ell a$ for some $j \geq 1$, $\ell \geq 1$. 
As $\pal{a^iba^{i+j} bb} = 3$, we have $\ell = 1$.
As $\pal{a^iba^{i+j} ba^mb} = 3$ when $m < i$, we deduce that 
$u = a^iba^{i+j}ba^i$. Hence $\nextt(a^i b a^{i+1}) = \{ a^iba^{i+j}ba^i \mid j \geq 1\}$.

\item Let $i \geq 1$, $j \geq 1$, $k \geq 1$.
Observe first that
$a^i (b a^{i+j})^{k+1}ba^i \in \nextt(a^i (b a^{i+j})^kba^i)$.
Assume by contradiction that there exists an element $u$ in the set
$\nextt(a^i (b a^{i+j})^kba^i)\setminus\{ a^i (b a^{i+j})^{k+1}ba^i \}$.
It has a prefix in the form
$p_1 = a^i (b a^{i+j})^{k}ba^ia^\ell b$ with $\ell \neq j$ 
or in the form
$p_2 = a^i (b a^{i+j})^{k}ba^{i+j}b^n$ with $n \geq 2$.
or in the form
$p_3 = a^i (b a^{i+j})^{k}ba^{i+j}ba^mb$ with $0 \leq m < i$.
Observe that $\pal{p_2} = 3$, $\pal{p_3} = 3$ and, $\pal{p_1} \leq 2$ only if $\ell = 0$.
In this latter case, $u$ has $p_1 = a^i (b a^{i+j})^{k}ba^i b$ as a prefix.
Consequently for an integer $n \geq 1$, $u$ has also a prefix $p_4 = a^i (b a^{i+j})^{k}ba^i b^na$.
We have $\pal{p_4} \leq 2$ only if $n = 1$ and $j = 1$: $p_4 = a^i (b a^{i+1})^{k}ba^i ba$.
Observe that $\pal{p_4a} = \pal{a^i (b a^{i+1})^{k}ba^i baa} = 3$.
Thus for an integer $\alpha \geq 1$, 
$p_5 = a^i(ba^{i+1})^kba^ibab^\alpha a$ is a prefix of $u$.
Now $\pal{p_5} \leq 2$ only if $\alpha = 1$ and $i = 1$: $p_5 = (aba)^ka(ba)^3$.
Observe that for any $\gamma \geq 3$, 
$\pal{(aba)^ka(ba)^\gamma} = \pal{a(baa)^k(ba)^\gamma b} = 2$ and
$\pal{a(baa)^k(ba)^\gamma bb}  = \pal{a(baa)^k(ba)^\gamma a}  = 3$. 
Thus $u$ must have all words $(aba)^ka(ba)^\gamma$ as prefixes (for any $\gamma \geq 3$).
We get a contradiction with the finiteness of $u$.
Hence $\nextt(a^i (b a^{i+j})^kba^i) = \{a^i (b a^{i+j})^{k+1}ba^i\}$.

\item Let $k \geq 2$: $(a(ba)^k)^2 \in \nextt(a(ba)^ka)$. 
Let $u$ be an element of the set $\nextt(a(ba)^ka)\setminus \{ (a(ba)^k)^2\}$. 
It has a prefix $p_1 := a(ba)^{k}(ab)^i b$ with $0 < i \leq k$
or $p_2 = a(ba)^k(ab)^iaa$ with $0 \leq i < k$.
Observe that $\pal{p_1} =3$. Also $\pal{p_2} =3$ except when $i = 0$.
It follows that $u$ has a prefix $p_3 = a(ba)^ka^\alpha b^\beta a$ with $\alpha \geq 2$, $\beta \geq 1$.
This is not possible as $\pal{p_3} = 3$.
Hence $\nextt(a(ba)^ka) = \{ (a(ba)^k)^2\}$. 

\item Let $k \geq 2$ and $\alpha \geq 2$.
First  $(a(ba)^k)^{\alpha+1} \in \nextt((a(ba)^k)^\alpha)$.
Let $u$ be an element of $\nextt((a(ba)^k)^\alpha) \setminus \{(a(ba)^k)^{\alpha+1}\}$.

This word $u$ has a prefix in the form $p_1 = (a(ba)^k)^\alpha (ab)^ib$ with $0 \leq i \leq k$ or
in the form $p_2 = (a(ba)^k)^\alpha (ab)^i aa$ with $0 \leq i < k$.
We have $\pal{p_1} \leq 2$ or $\pal{p_2} \leq 2$ only if $i = 0$.
In this case $u$ has a prefix $p_3 = (a(ba)^k)^\alpha b^\beta a$ with $\beta \geq 1$ or
$p_4 = (a(ba)^k)^\alpha a^\beta b$ with $\beta \geq 2$.
Observe that $\pal{p_4} = 3$ and, $\pal{p_3} \leq 2$ only if $\beta = 1$.
Thus $u$ has $p_3 = (a(ba)^k)^{\alpha-1}a(ba)^{k+1}$ as a prefix.
Let $\gamma \geq 1$.
Observe that $\pal{(a(ba)^k)^{\alpha-1}a(ba)^{k+\gamma}} = 2$ and
$\pal{(a(ba)^k)^{\alpha-1}a(ba)^{k}(ba)^\gamma b} = 2$ while
$\pal{(a(ba)^k)^{\alpha-1}a(ba)^{k+\gamma}a} = 3$ and
$\pal{(a(ba)^k)^{\alpha-1}a(ba)^{k}(ba)^\gamma bb} = 3$.
We get a contradiction with the finiteness of $u$.
Hence $\nextt((a(ba)^k)^\alpha) = \{(a(ba)^k)^{\alpha+1}\}$.

\end{enumerate}
\end{proof}

\begin{proof}[Proof of Lemma~\ref{L:approche2bis}]

The proof of the \emph{if part} just needs to verify that each prefix of a word in one of the four forms 
can be decomposed in one or two palindromes.
This is done straightforwardly as these prefixes are words in one of the following forms:
\begin{itemize}
\item $a^i$, $i \geq 0$;
\item $a^i(ba^{i+j})^k ba^ia^\ell$, $i \geq 1$, $k \geq 0$, $j \geq 0$, $\ell \geq 0$;
\item $a^{i-\ell}a^\ell(ba^{i+j})^k ba^\ell$, $i \geq 1$, $0 \leq \ell \leq i$, $j \geq 0$, $k \geq 0$;
\item $a^{i-\ell}a^\ell(b^ja^{i})^k b^ja^\ell$, $i \geq 1$, $\ell \leq i$, $j \geq 1$, $k \geq 0$;
\item $(a^ib^j)^k a^i b^\ell$, $i \geq 1$, $j \geq 1$, $k \geq 0$, $\ell \leq j$,
\item $((ab)^ia)^k(ab)^\ell a$,   $i \geq 0$, $k \geq 0$, $\ell \geq 0$;
\item $(ab)^{i-\ell}a.(ba)^\ell((ab)^ia)^k(ab)^\ell$, $i \geq 1$, $0 \leq \ell \leq i$, $k \geq 0$.
\end{itemize}

We now prove the only if part.
Let $\bw$ be an infinite word in $\BPLP(2) \cap \SWHIMPP$ beginning with the letter $a$. 
By Lemma~\ref{L:size_alphabet}, $\bw$ contains at most two different letters. 
If $\bw$ is written only on one letter, $\bw = a^\omega$. 
Assume $\bw$ is written on two letters.
Let $b$ be the second letter.
For some integer $i \geq 1$, $\bw$ begins with $a^i b$.
By Items~\ref{item1} and \ref{item2} of Lemma~\ref{L:approche2}, for some integer $j \geq 1$, 
$\bw$ begins with $a^ib^ja^i$. 
Assume $\bw \neq (a^ib^j)^\omega$.
By Item~\ref{item4} of Lemma~\ref{L:approche2}, $\bw$ begins with $a^i(b^ja^i)^ka$ for some $k \geq 1$.
By Item~\ref{item3} of Lemma~\ref{L:approche2},  $j =1$.
Assume $k =1$. 
By Item~\ref{item6} of Lemma~\ref{L:approche2} there exists an integer $j' \geq 1$ such that $\bw$ begins with $a^ib^{i+j'}ba^i$. 
Thus by Item~\ref{item7} of Lemma~\ref{L:approche2}, $\bw = (a^iba^{j'})^\omega$.
Assume from now on that $w$ begins with $a^i(ba^i)^ka$ with $k \geq 2$.
By Item~\ref{item5} of Lemma~\ref{L:approche2}, $i = 1$.
By Items~\ref{item8} and \ref{item9}  of Lemma~\ref{L:approche2},
$\bw = (a(ba)^k)^\omega$.
\end{proof}

Maybe the main interest of the previous proof is the idea of studying the links between successive palindromic prefixes of the considered infinite words. 
Nevertheless determining $\BPLF(3)$ or $\BPLP(3)$ seems much more difficult due to a combinatorial explosion even when restricted to binary words.
Contrarily to words in $\BPLF(2)$, words in $\BPLP(3)$ (and so in $\BPLF(3)$) contains ternary words. For instance the word $(abac)^\omega$, with $a$, $b$ and $c$ three different letters, belong to $\BPLP(3)$.

\section{\label{sec:greedypal}Greedy palindromic lengths}

We now introduce greedy palindromic lengths. In this section we provide generalities on these notions.
In the next sections, we will show that any infinite word with left or right bounded greedy palindromic lengths is ultimately periodic. This proves Conjecture~\ref{conj1} in a special case and reinforces all conjectures studied in this paper.

\subsection{Definition and examples}

Let us define inductively the \textit{left greedy palindromic length} of a word $w$ by: $\lgpal{\varepsilon} = 0$, and, 
$\lgpal{w} = 1  + \lgpal{u}$ when $w \neq \varepsilon$ and $w = \pi u$ with $\pi$ the longest palindromic prefix of $w$. For instance, $\lgpal{abaa} = 2$ and $\lgpal{abaab} = 3$.
Similarly, we define the \textit{right greedy palindromic length} $\rgpal{w}$ considering at each step the longest palindromic suffix: 
$\rgpal{abaa} = 3$ and $\rgpal{abaab} = 2$.
An important difference between greedy palindromic lengths and the palindromic length is that the definition implies a unique decomposition. For a finite word $w$, we say that $(\pi_1, \ldots, \pi_k)$ is the left greedy palindromic decomposition of $w$, if $w = \pi_1 \cdots \pi_k$ and, moreover, for all $i$, $1 \leq i \leq k$, $\pi_i$ is the longest palindromic prefix of $\pi_i \cdots \pi_k$. The right greedy palindromic decomposition of a word can be defined similarly.

In order to provide examples and following the idea of Section~\ref{sec:around}, let us consider an infinite word such that, for all prefixes $p$, $\lgpal{p} \leq 2$. 
For some letters $a$ and $b$ and integers $n \geq 0$, $k \geq 1$, 
this word is in the form $a^nb^\omega$, 
$(ab^k)^\omega$,
$(ab^k)^na^\omega$
or $(ab^k)^nab^\omega$ (only words in the form $b^\omega$ or $(ab^k)^\omega$ belong to $\SWHIMPP$).
Let us prove this claim.
First all these forms are suitable as prefixes of such words are, for some integer $k \geq 1$, $n \geq 1$, $i \geq 0$, in the form $a^n$, $a^nb^k$, $(ab^k)^naa^i$ or $(ab^k)^nab^i$. Moreover we have:
$\lgpal{a^n} = 1$; 
$\lgpal{a^nb^k} = 2$;
$\lgpal{(ab^k)^na} = 1$;
$\lgpal{(ab^k)^naa^i} = 2$ if $i \geq 1$;
$\lgpal{(ab^k)^nab^i} = 2$ if $i \geq 1$.
No other word is suitable. 
Indeed if a word is not in these forms, 
for some integers $k \geq 1$, $n \geq 1$, 
it has a prefix in the form $a^ib^ka$ with $i \geq 2$,
$(ab^k)^na^ib$ with $i \geq 2$ or $(ab^k)^nab^ia$ with $i \neq k$. Moreover
$\lgpal{a^ib^ka} = 3$ if $i \geq 2$;
$\lgpal{(ab^k)^na^ib} = 3$ if $i \geq 2$;
$\lgpal{(ab^k)^nab^ia} = 3$ if $i \neq k$.

Now let us consider an infinite word such that, for all prefixes $p$, $\rgpal{p} \leq 2$. 
For some integers $n \geq 1$ and $k \geq 1$, this word is in the form 
$a^\omega$, 
$a^nb^\omega$,
$a^n(ba^\ell)^\omega$ with $\ell < n$, 
$(a^nb)^\omega$
or
$(a^nb)^kb^\omega$
(only $a^\omega$ and $(a^nb)^\omega$ belong to $\SWHIMPP$).
Let us prove this claim.
First all these forms are suitable. 
Indeed prefixes of such words are, for some integers $n \geq 0$, $k \geq 1$, $\ell \geq 0$,
in one of the following forms:
$a^n$;
$a^nb^k$;
$a^n(ba^\ell)^iba^j$ with $\ell < n$, $i \geq 0$, $0 \leq j \leq \ell$; 
$(a^nb)^ka^\ell$  with $\ell \leq n$;
$(a^nb)^kb^\ell$. 
We have for $n \geq 1$ and $k \geq 1$: 
$\rgpal{a^n} = 1$;
$\rgpal{a^nb^k} =2$; 
$\rgpal{a^n(ba^\ell)^iba^j} = 2$ when $\ell < n$, $i \geq 0$, $0 \leq j \leq \ell$; 
$\rgpal{(a^nb)^ka^\ell} = 1$ if $\ell = n$,  $2$  if $\ell < n$;
 $\rgpal{(a^nb)^kb^\ell} = 2$.
No other word is suitable. 
Indeed if a word is not in this form, for some integers $k \geq 2$, $n \geq 1$, $\ell \geq 1$ 
it has a prefix in one of the following forms: 
$a^nb^ka^\ell$ with $\ell \geq n+1$;
$a^nb^ka^\ell b$ with $k \geq 2$ and $\ell \neq 0$;
$a^n(ba^\ell)^iba^j$ with $j \geq n+1$ and $\ell \leq n$;
$a^n(ba^\ell)^iba^jb$ with $i \geq 1$, $j \neq \ell$, $j > 0$ and $\ell \neq n$;
$a^n(ba^n)^iba^jba$ with $i \geq 1$ and $j < n$;
$a^n(ba^n)^iba^jbb$ with $0 < j < n$. 
Observe that the right greedy palindromic length is 3 for all these words.

We have already mentioned that the palindromic length of a finite word of length $n$ can be computed in time $O(n log(n))$ \cite{Fici_Gagie_Karkkainen_Kempa2014JDA,Shur_Rubinchik2015Iwoca}. We claim that greedy palindromic lengths of such a word can be computed in time $O(n)$.
As the left greedy palindromic length of a finite word $w$ is the right greedy palindromic length of the mirror word $\tilde{w}$, we just have to explain this for the right greedy palindromic length.
In the article \cite{Groult_Prieur_Richomme2010IPL}, it was explained how to compute in linear time, for a word $w$, an array LPS that stores for each $i$, $1 \leq i \leq |w|$, the length of the longest palindromic
suffix ending at position $i$:
$${\text LPS}[i] = \max\{\ell \mid w[i - \ell+1..i] \text{ is a palindrome }\}$$
Applying next algorithm after computing this array LPS allows to compute the right greedy palindromic length of $w$ (the result is store in variable RGPal). 
The whole computation takes an $O(|w|)$ time.

\begin{quote}
\begin{tabbing}
RGPal $\leftarrow 0$\\
$i$ $\leftarrow$ $|w|$\\
while $i>0$ do:\\
\hspace*{1cm} \= RGPAL $\leftarrow$ RGPAL $+ 1$\\
\> $i$ $\leftarrow$ i - LPS[i]
\end{tabbing}
\end{quote}

\subsection{Links with palindromic length}

The following relation follows directly the definitions of greedy palindromic lengths.

\begin{prprt}
\label{P:property1}
For any word $u$, $|u|_{pal} \leq \min(\lgpal{u}, \rgpal{u})$.
\end{prprt}

Next example, provided to us by P.~Ochem, shows that the value $\min(\lgpal{u},$ $\rgpal{u})-|u|_{pal}$ can be arbitrarily large. Let $m_n$ be the $n^{\rm th}$ term of the Multibonacci sequence of words defined over the set of integers viewed as letters by: $m_1 = 1$, $m_{n+1} = m_n(n)m_n$ ($m_2 = 121$, $m_3 = 1213121$, \ldots). All words $m_n$ end with the letter $1$ and readers can verify that $\pal{m_n1^{-1}} = 2$ while $\lgpal{m_n1^{-1}} = 2n-2 = \rgpal{1^{-1}m_n}$.
Now let $a$ and $b$ be two symbols that are not integers. For $n \geq 1$, 
let $M_n = 1^{-1}m_n ab m_n 1^{-1}$. We have 
$\lgpal{M_n} = 2+ \lgpal{1^{-1}m_n}+ \lgpal{m_n 1^{-1}} = 2n+2 = 
2+ \rgpal{1^{-1}m_n}+ \rgpal{m_n 1^{-1}} = 
\rgpal{M_n}$. Moreover $|M_n|_{pal} = 6$.

A similar example over a binary alphabet can be found in the online proceedings of Conference ``Journ\'ees Montoises d'Informatique Th\'eorique 2014".
It is also possible to recode previous example applying a morphism on word $m_n$ (as for instance, morphism $f_n$ defined by $f_n(i) = a^{n+1-i}b^{2i}a^{n+1-i}$).

As a consequence of Property~\ref{P:property1}, if for an infinite word $\bw$ there exists an integer $K$ such that $\lgpal{p} \leq K$ for any prefix $p$ of $\bw$, then also $|p|_{pal} \leq K$. We have the following stronger property (let recall that the existence of a similar property when considering the palindromic length instead of the left greedy palindromic length is an open problem):

\begin{prprt}
\label{P:property2}
If for an infinite word $\bw$ there exists an integer $K$ such that $\lgpal{p} \leq K$ for any prefix $p$ of $\bw$, then also $|u|_{pal} \leq 2K$ for any factor $u$ of $\bw$. 
\end{prprt}

\begin{proof}
Let $u$ be a factor of $\bw$ and let $p$ be such that $pu$ is a prefix of $\bw$. 
Let $\pi_1$, \ldots, $\pi_k$ be the palindromes such that $(\pi_1, \ldots, \pi_k)$ is the left greedy palindromic 
decomposition of $pu$: $k = \lgpal{pu} \leq K$.
Let $x$, $y$ be the words and $i$ be the integer such that $y \neq \varepsilon$, 
$\pi_i = xy$ and $u = y \pi_{i+1} \cdots \pi_k$. 
The word $\tilde{y}$ is a prefix of the palindrome $\pi_i$. 
Hence $\pi_1\cdots\pi_{i-1}\tilde{y}$ is a prefix of $\bw$. 
Let $\pi_1'$, \ldots, $\pi_\ell'$ be the palindromes such that $(\pi_1', \ldots, \pi_\ell')$ is the left greedy palindromic 
decomposition of $\pi_1\cdots\pi_{i-1}\tilde{y}$: $\ell = \lgpal{\pi_1\cdots\pi_{i-1}\tilde{y}} \leq K$.
Palindromes $\pi_j$ and $\pi_j'$ are the longest palindromic prefixes of respectively $\pi_j \cdots \pi_k$ and 
$\pi_j' \cdots \pi_{\ell}'$. It follows that $i-1 < \ell$,  
$\pi_1 = \pi_1'$, \ldots, $\pi_{i-1} = \pi_{i-1}'$ and $\tilde{y} = \pi_i' \cdots \pi_{\ell}'$. 
Hence $u = \pi_{\ell}' \cdots \pi_i' \pi_{i+1}\cdots \pi_k$ is the product of at most $2K$ palindromes.
\end{proof}

An interesting question is whether Property~\ref{P:property2} is still true when considering right greedy palindromic length instead of left greedy palindromic length. This is an open problem.

\section{Right greedy palindromic length}

In this section, we prove Conjecture~\ref{conj1} in the special case where the palindromic length is replaced with the right greedy palindromic length.
For an infinite word $\bw$, let $\maxrgpalpref{\bw}$ denote the supremum of the set $\{\rgpal{p} \mid p \text{ prefix of } \bw\}$.

\begin{thrm}
\label{T:RGPal}
Let $\bw$ be a non ultimately periodic infinite word.
Then $\maxrgpalpref{\bw}$ is infinite, that is, there exist prefixes of $\bw$ with arbitrarily large
right greedy palindromic lengths.
\end{thrm}

\begin{proof}
Assume first that the word $\bw$ has no suffix in $\SWHIMPP$.
By Lemma~\ref{L:ex2.1(2)}, the palindromic lengths of prefixes of $\bw$ are unbounded.
Property~\ref{P:property1} implies that the right greedy palindromic lengths of prefixes of $\bw$ are unbounded.

Assume from now on that
$\bw = u \bw'$ with $u$ some finite word and $\bw' \in \SWHIMPP$.

Let $(\pi_1, \ldots, \pi_n)$ be 
the right greedy palindromic decomposition of the word $u$.
Let recall that this means that $n = \rgpal{u}$, $u = \pi_1 \cdots \pi_n$ and, 
for $i$, $1 \leq i \leq n$, $\pi_i$ is the longest palindromic suffix of 
$\pi_1 \cdots \pi_i$.

More generally, assume we have found a prefix $p$ of $\bw$ with $|p| \geq |u|$ such that its right greedy palindromic decomposition is in the form $(\pi_1, \ldots, \pi_k)$ with $k = \rgpal{p} \geq n$.
Here we mean that this decomposition begins with $(\pi_1, \ldots, \pi_n)$, the right greedy palindromic decomposition of $u$. In three steps, we are going to show how to find a palindrome $\pi_{k+1}$ such that $p\pi_{k+1}$ is a prefix of $\bw$ and $(\pi_1, \ldots, \pi_{k+1})$ is the right greedy palindromic decomposition of $p\pi_{k+1}$. Thus $\rgpal{p\pi_{k+1}} = k+1$. The proof of the theorem follows this construction by induction.

\textbf{Step 1:} \textit{Marking a first occurrence of a new factor $v$}.

Let $x$ be the word of length $|\pi_1\cdots \pi_k|+1$ such that $\pi_1\cdots \pi_k x$ is a prefix of $\bw$.
This word $\pi_1\cdots \pi_k x$ contains at most $|x|$ different factors of length $|x|$.
As $\bw$ is not ultimately periodic, by the celebrated Morse-Hedlund theorem 
(see, \textit{e.g.}, \cite[Th. 10.2.6]{AlloucheShallit2003book}), the word 
$\bw$ has at least $|x|+1$ different factors of length $|x|$.
Hence there exists a word $v$ of length $|x|$ that does not occur in $\pi_1\cdots \pi_k x$.

\textbf{Step 2:} \textit{Constructing $\pi_{k+1}$}.

Let $y$ be the smallest word such that $\pi_1\cdots \pi_k y v$ is a prefix of $\bw$ (remarks: the word $v$ may overlap the word $x$; as $x \neq v$, $|y| \geq 1$).
Let also $\bw''$ denote the word such that $\bw = \pi_1\cdots \pi_k y v \bw''$.
As $|\pi_1\cdots \pi_k| \geq |u|$,
$yv\bw''$ is a suffix of $\bw'$.
As $\bw' \in \SWHIMPP$, $yv\bw''$ also belongs to $\SWHIMPP$.
Hence there exists a palindrome $\pi_{k+1}$ which is a prefix of $yv\bw''$ of length at least $|yv|$.

\textbf{Step 3:} \textit{The word $\pi_{k+1}$ is the longest palindromic suffix of $\pi_1\cdots \pi_k\pi_{k+1}$}.

Let $\pi$ be the longest palindromic suffix of $\pi_1\cdots \pi_{k+1}$ ($|\pi| \geq |\pi_{k+1}|$) and let
$z$ be the word such that $\pi_1\cdots \pi_{k+1} = z \pi$. 
As $\pi_{k+1}$ ends with $\tilde{v}\tilde{y}$, the palindrome $\pi$ also ends with $\tilde{v}\tilde{y}$.
Thus the word $zyv$ is a prefix of $z \pi$ and so of $\bw$. 
By construction, $\pi_1\cdots \pi_k y$ is the smallest word such that 
$\pi_1\cdots \pi_k y v$ is a prefix of $\bw$.
Hence,  $|z| \geq |\pi_1\cdots \pi_k|$ and so $|\pi| \leq |\pi_{k+1}|$: $\pi = \pi_{k+1}$.
\end{proof}

\section{\label{sec:LGPAL}Left greedy palindromic length}

In this section, we prove next result, that is, Conjecture~\ref{conj1} in the special case where the palindromic length is replaced with the left greedy palindromic length. 
For an infinite word $\bw$, we let $\maxlgpalpref{\bw}$ denote the supremum of the set $\{\lgpal{p} \mid p \text{ prefix of } \bw\}$.

\begin{thrm}
\label{T:LGPal}
Let $\bw$ be a non ultimately periodic infinite word.
Then $\maxlgpalpref{\bw}$ is infinite, that is, the left greedy palindromic lengths of prefixes of $\bw$ are unbounded.
\end{thrm}

The proof of this result appears to be much more difficult than the proof of Theorem~\ref{T:RGPal}, its analogue with the right greedy palindromic length. 
Section~\ref{subsec:mainidea} explains the main idea of our strategy and the difficulties to apply it. This section also provides intermediary tools.
Section~\ref{subsec:proofThLGPAL} contains the proof of Theorem~\ref{T:LGPal}.

Before going further, let us explain that we can reduce the proof of this theorem to words in $\SWHIMPP$, that is, to words beginning with infinitely many palindromes. Let $\bw$ be a non ultimately periodic infinite word.
By Property~\ref{P:property1} and Lemma~\ref{L:ex2.1(2)}, if $\bw$ has no suffix in $\SWHIMPP$, then the palindromic lengths of prefixes of $\bw$ are unbounded.

Thus we can assume that $\bw$ has a suffix in $\SWHIMPP$.
Next lemma shows that, when proving Theorem~\ref{T:LGPal}, one can replace $\bw$ by one of its
suffixes in $\SWHIMPP$.

\begin{lmm}
\label{L:reduction_lgpl}
Let $\bw$ be an infinite word having a suffix in $\SWHIMPP$.
There exist palindromes $\pi_1, \ldots, \pi_k$ and an infinite word $\bw'$ in $\SWHIMPP$ such that $\bw = \pi_1 \cdots\pi_k \bw'$ and, for all $1 \leq i \leq k$,
$\pi_i$ is the longest palindromic prefix of $\pi_i\cdots \pi_k \bw'$. Moreover for any prefix $p$ of $\bw''$, 
$\lgpal{p} = \lgpal{\pi_1\cdots \pi_k p} - k$.
\end{lmm}

\begin{proof}
Observe that a word $\bw$ has a (unique) largest palindromic prefix
if and only if
it has finitely many palindromic prefixes, that is,
it does not belong to $\SWHIMPP$.

When $\bw \not\in \SWHIMPP$, the word $\bw$ can be decomposed 
$\bw = \pi_1\bw_1$ with $\pi_1$ the largest palindromic prefix of $\bw$ and $\bw_1$ an infinite word.
Moreover for any prefix of $\bw_1$,
$\lgpal{p} = \lgpal{\pi_1p}-1$.

Iterating this process, one can find palindromes $\pi_1$, \ldots, $\pi_k$
and infinite words $\bw_1$, \ldots, $\bw_k$ such that:
$\bw = \pi_1\cdots \pi_i \bw_i$ for all $i$, $1 \leq i \leq k$;
$\bw_i \not\in \SWHIMPP$, for all $i$, $1 \leq i \leq k-1$;
$\pi_j$ is the longest palindromic prefix of $\pi_j \cdots \pi_i\bw_i$, 
for all $i$, $j$, $1 \leq j \leq i \leq k$;
for any prefix $p$ of $\bw_i$, $\lgpal{p} = \lgpal{\pi_1\cdots \pi_ip}-i$.
The iteration ends when $\bw_k$ belongs to $\SWHIMPP$. 
This must occur as $\bw$ has a suffix in $\SWHIMPP$ and, as any suffix of a word in $\SWHIMPP$ belongs to $\SWHIMPP$. Taking $\bw' = \bw_k$ ends the proof of the lemma.
\end{proof}

\subsection{Tools}

For a word in $\SWHIMPP$, we will need to consider the 
\textit{length increasing sequence of palindromic prefixes} of $\bw$, 
that is the sequence of palindromic prefixes  $(\pi_n)_{n \geq 1}$ of $\bw$ such that all palindromic prefixes of $\bw$ occur in the sequence and $(|\pi_n|)_{n \geq 0}$ is (strictly) increasing. We will use Lemma~\ref{L:tool2} that characterizes periodicity in the context of words in $\SWHIMPP$. It will be useful in our last section and may be important in a more general context. For the proof, we need next result.

\begin{lmm}[\mbox{\cite[Th. 4]{BrlekHamelNivatReutenauer2004IJFCS}}]
\label{L:ex2.1(3)}
An infinite periodic word $w^\omega$ with $w$ primitive contains infinitely many palindromes if and only if $w$ is the product of two palindromes.
\end{lmm}

\begin{lmm}
\label{L:tool2}
Let $\bw$ be an infinite word in $\SWHIMPP$.
Let $(\pi_i)_{i \geq 0}$ be its sequence of length increasing palindromic prefixes. 
The sequence $(|\pi_{i+1}|-|\pi_i|)_{i \geq 0}$ is not decreasing.
Moreover this sequence is bounded if and only if $\bw$ is periodic.
\end{lmm}

\begin{proof}
Assume by contradiction that there exists an integer $j$ such that 
$|\pi_{j+2}|-|\pi_{j+1}| < |\pi_{j+1}|-|\pi_{j}|$.
By \cite[Lem. 2]{FridPuzyninaZamboni2013AAM}
it is known that if $u = v s$ for two palindromes $u$ and $v$ and a word $s$, then $|s|$ is a period of $u$ 
(here an integer $p$ is a \textit{period} of a finite word $w$, 
if for all integers $i$ between $1$ and $|w|-p$, $w[i] = w[i+p]$).
Let $p = |\pi_{j+2}|-|\pi_{j+1}|$.
From the previous recalled result, $p$ is a period of $\pi_{j+2}$ and so also of $\pi_{j+1}$.
As $p < |\pi_{j+1}| - |\pi_{j}| < |\pi_{j+1}|$, we get that the prefix of length $|\pi_{j+1}|-p$ of $\pi_{j+1}$ is equal to the suffix of same length of this word. Let $\pi$ be this word. As $\pi_{j+1}$ is a palindrome, also $\pi$ is a palindrome. 
Observe that $|\pi| = |\pi_{j+1}| - p > |\pi_{j+1}| - (|\pi_{j+1}| - |\pi_{j}|) = |\pi_{j}|$. Hence $|\pi_{j}| < |\pi| < |\pi_{j+1}|$: we have found a contradiction with the definition of the sequence $(\pi_i)_{i \geq 0}$.

When the sequence $(|\pi_{i+1}|-|\pi_i|)_{i \geq 0}$ is bounded, it is ultimately periodic.
By \cite[Lem. 2]{FridPuzyninaZamboni2013AAM}, its ultimate value $p$ is a period of all palindromic sufficiently large prefixes $\pi_i$ of $\bw$. Consequently, $\bw$ has period $p$. 

Assume $\bw$ is periodic. By Lemma~\ref{L:ex2.1(3)}, $\bw = (uv)^\omega$ with $u$ and $v$ two palindromes. 
The sequence $((uv)^ju)_{j \geq 1}$ is a subsequence of $(\pi_i)_{i \geq 0}$. So the sequence $(|\pi_{i+1}|-|\pi_i|)_{i \geq 0}$ is bounded.
\end{proof}

\subsection{\label{subsec:mainidea}Main ideas on the proof of Theorem~\ref{T:LGPal}}

The main idea of our proof of Theorem~\ref{T:LGPal} is to find, when possible, situations like this of next lemma.
When a word $p$ is a prefix of a word $u$, we let $p^{-1}u$ denote the suffix $s$ of $u$
such that $u = ps$.
When a word $s$ is a prefix of a word $u$, we let $us^{-1}$ denote the prefix $p$ of $u$
such that $u = ps$.

\begin{lmm}
\label{L:main idea}
Let $\bw \in \SWHIMPP$ and $(\pi_i)_{i \geq 1}$ be the length increasing sequence of palindromic prefixes of $\bw$.
Assume there exist infinitely many integers $i$ for which there exists a palindrome $p_i$ which 
is a factor of $\pi_{i+1}$ but does not occur in nor does not overlap the prefix $\pi_i$ of $\pi_{i+1}$.
The left greedy palindromic lengths of prefixes of $\bw$ are not bounded.
\end{lmm}

\begin{proof}
Let $i$ be an integer and $p_i$ be a palindrome such that $p_i$ is a factor of $\pi_{i+1}$ 
that does not occur in nor does not overlap the prefix $\pi_i$ of $\pi_{i+1}$.
There exist words $x$ and $y$ such that $\pi_i^{-1}\pi_{i+1} = xp_i y$.
The word $\tilde{y}\widetilde{p_i}\tilde{x} = \tilde{y}p_i \tilde{x}$ is a prefix of $\pi_{i+1}$.
By hypothesis, $|\tilde{y}| \geq |\pi_i|$ (otherwise $p_i$ would be a factor of the prefix $\pi_i$ of $\pi_{i+1}$, or would overlap it).
Hence $y = z \pi_i$ for a word $z$ ($\tilde{y} = \pi_i \tilde{z}$).
Let $\pi = x p_i z$. We have:
$$\pi_{i+1} = \pi_i \pi \pi_i$$

We now prove that $\pi$ is the longest palindromic prefix of $\pi \pi_i$.
Without loss of generality, 
we  assume that there is exactly one occurrence of $p_i$ in $\pi_i x p_i$.
Let $\pi'$ be the longest palindromic prefix of $\pi \pi_i$, 
and let $s$ be the word such that $\pi' s = \pi \pi_i$.
As $|\pi'| \geq |\pi|$ and as $xp_i$ is a prefix of $\pi$,
the word $p_i \tilde{x}$ is a suffix of $\pi'$.
Consequently $\tilde{s} x p_i$ is a prefix of $\pi_{i+1}$.
By choice on the occurrence of $p_i$ at the beginning of the paragraph, 
$|sxp_i| \geq |\pi_i x p_i|$.
By definition of $\pi'$, we also have 
$|s| \leq |\pi_i|$. 
Thus $|s| = |\pi_i|$ and $\pi' = \pi$.
That $\pi$ is the longest palindromic prefix of $\pi\pi_i$
implies that, for any proper prefix $p$ of $\pi_i$,
we have $\lgpal{\pi_i\pi p} = 2 + \lgpal{p}$ (let recall that by definition of the sequence $(\pi_i)_{i \geq 1}$, $\pi_i$ is the longest prefix of $\pi_{i+1}$).

As there exist infinitely many couples $(i, p_i)$, 
from what precedes, 
we can construct an infinite sequences $(u_i)_{i \geq 1}$ of prefixes of $\bw$ such that,
for all $i \geq 1$, there exists an integer $j \geq i$ and a palindrome $v_i$ such that
$u_{i+1} = \pi_j v_i u_i$,
$u_{i+1}$ is a prefix of $\pi_{j+1} = \pi_j v_i \pi_j$,
$v_j$ is the longest palindromic prefix of $v_i \pi_j$ and $\lgpal{u_{i+1}} =  2 + \lgpal{u_i}$.
Left greedy palindromic lengths of prefixes of $\bw$ are not bounded.
\end{proof}

Next result shows an example of use of previous lemma.
This result is Theorem~\ref{T:LGPal} in the particular case of a word over an infinite alphabet.

\begin{prpstn}
\label{P:lgpal_infinite_alphabet}
For any infinite word $\bw$ over an {\em infinite} alphabet, if $\bw$ has 
infinitely many palindromic prefixes, then
the set $\{\lgpal{p} \mid p \mbox{ prefix of } \bw \}$ is unbounded.
\end{prpstn}

\begin{proof}
Let $(\pi_n)_{n \geq 1}$ be the length increasing sequence of palindromic prefixes of $\bw$.
As the alphabet is infinite, there exist infinitely many integers $i$ such that $\pi_{i+1}$ contains a letter $a_i$ that does not occur in $\pi_i$.
As letters are palindromes, by Lemma~\ref{L:main idea}, $\maxlgpalpref{\bw}$ is infinite.
\end{proof}

In next example, the situation of Lemma~\ref{L:main idea} does not hold.
So, our approach will have to be adapted.
This will lead to the study of three different cases at the end of the proof of Theorem~\ref{T:LGPal} in next section.

\begin{xmpl}
Let 
$\pi_0 = aba$; for all $n \geq 1$, $\pi_{n} = (\pi_{n-1} a^n)^{(+)} = \pi_{n-1} a^{n-1}\pi_{n-1}$ 
(where $^{(+)}$ denotes the palindromic closure; 
for a word $u$, $u^{(+)}$ is the smallest palindrome having $u$ has a prefix):

$\pi_1 = abaaba$

$\pi_2 = abaabaaabaaba$

$\pi_3 = abaabaaabaabaaaabaabaaabaaba$

$\pi_4 = abaabaaabaabaaaabaabaaabaabaaaaabaabaaabaabaaaabaabaaabaaba$

For all $n \geq 1$, the longest palindromic prefix of $\pi_n^{-1} \pi_{n+1}$ is the word $a^n$. This palindrome overlaps the  last occurrence of $\pi_n$ and its prefix $a^{n-1}$ does not correspond to the first occurrence of this word $a^{n-1}$ in $\pi_{n+1}$.
\end{xmpl}

To end with the general ideas of the proof of Theorem~\ref{T:LGPal}, 
let us mention that next lemma will allow to consider palindromes in the form $(uv)^ku$ for words $p_i$ when using Lemma~\ref{L:main idea}.

\begin{lmm}
\label{X1}
Let $\bw$ be a non ultimately periodic word whose left greedy palindromic lengths of prefixes
are bounded. 
There exist palindromes $u$ and $v$ with $uv$ primitive (in particular $uv \neq \varepsilon$) such that, 
for all $k \geq 0$, $(uv)^k u$ is a factor of $\bw$.
\end{lmm}

\begin{proof}
Assume first that there exist palindromes $u$ and $v$ with $uv \neq \varepsilon$  such that, 
for all $k \geq 0$, $(uv)^k u$ is a factor of $\bw$. Let us show that we can assume $uv$ primitive.
Assume $uv$ is not primitive and let $z$ be its primitive root (that is the smallest word such that $uv$ is a power of $z$). There exist words $x$, $y$ 
and integers $\ell_1$, $\ell_2$ such that $z = xy$, $u = (xy)^{\ell_1}x$ and $v = y(xy)^{\ell_2}$. For any $k \geq 0$, $(uv)^ku = (xy)^{k(\ell_1+1+\ell_2)+\ell_1}x$. Moreover as $x$ is both a prefix and a suffix of the palindrome $u$, 
$x$ itself is a palindrome.  Similarly $y$ is a palindrome. Hence replacing $u$, $v$ with $x$, $y$ allows to assume that $uv$ is primitive.

We now show, under hypotheses $\maxlgpalpref{\bw}$ finite and $\bw$ non ultimately periodic,
the existence of palindromes $u$ and $v$ with $uv \neq \varepsilon$  such that, 
for all $k \geq 0$, $(uv)^k u$ is a factor of $\bw$.
We act by induction on $\maxlgpalpref{\bw}$.
First observe that if $\maxlgpalpref{\bw} = 1$,
	$\bw$ does not verify the hypotheses as it is periodic ($\bw = a^\omega$ with a letter).

Assume that $\bw \not\in \SWHIMPP$.
Then $\bw = \pi \bw'$ with $\pi$ the longest palindromic prefix of $\bw$.
For any prefix $p$ of $\bw'$, we have
$\lgpal{p} = \lgpal{\pi p} -1$.
Hence $\maxlgpalpref{\bw'} < \maxlgpalpref{\bw}$.
If $\bw'$ is ultimately periodic, 
by Lemma~\ref{L:ex2.1(3)}, 
$\bw' = p(uv)^\omega$ with $u$, $v$ two palindromes such that $uv \neq \varepsilon$ (and $p$ is a word).
Thus for all $k \geq 0$, $(uv)^ku$ is a factor of $\bw'$.
If $\bw'$ is not ultimately periodic, palindromes $u$ and $v$ exist by inductive hypothesis.

From now on, we assume that $\bw \in \SWHIMPP$.
Let $(\pi_i)_{i \geq 1}$ be the length increasing sequence of palindromic prefixes of $\bw$.
By Lemma~\ref{L:tool2}, the sequence $(|\pi_{i+1}|-|\pi_i|)_{i \geq 1}$ is not decreasing.
Moreover as $\bw$ is not periodic, this sequence is unbounded.
By K\"onig's lemma (see \cite[Prop. 1.2.3]{Lothaire2002book}), 
there exists an infinite word $\bw_1$ such that each of its prefixes $p$ is a proper prefix of 
$\pi_j^{-1}\pi_{j+1}$ for some $j \geq 1$.
By definition of the sequence $(\pi_i)_{i \geq 1}$, for any proper prefix $p$ of $\pi_j^{-1}\pi_{j+1}$,
$\lgpal{p} = \lgpal{\pi_jp} - 1$.
Hence $\maxlgpalpref{\bw_1} < \maxlgpalpref{\bw}$.
If $\bw_1$ is periodic, by Lemma~\ref{L:ex2.1(3)}, 
$\bw_1 = (uv)^\omega$ for some palindromes $u$, $v$ with $uv \neq \varepsilon$.
Thus for all $k \geq 0$, $(uv)^ku$ is a factor of $\bw_1$.
If $\bw_1$ is not periodic, palindromes $u$ and $v$ exist by inductive hypothesis. As factors of $\bw_1$ are factors of $\bw$, the lemma holds.
\end{proof}

\subsection{\label{subsec:proofThLGPAL}Proof of Theorem~\ref{T:LGPal}}

Let $\bw$ be a non ultimately periodic infinite word.
At the beginning of Section~\ref{sec:LGPAL}, 
we explained that, if $\bw$ has no suffix in $\SWHIMPP$,
then $\maxlgpalpref{\bw}$ is infinite. 
Thus we can assume that $\bw = u \bw'$ for some word $u$ and some infinite word $\bw'$ in $\SWHIMPP$.
Lemma~\ref{L:reduction_lgpl} shows that there exists a suffix $\bw''$ of $\bw'$ such that $\maxlgpalpref{\bw}$ is infinite if and only if $\maxlgpalpref{\bw''}$ is infinite. Thus from now on we assume that $\bw \in \SWHIMPP$.

We act by contradiction and so we assume that $\maxlgpalpref{\bw}$ is finite. 
Moreover without loss of generality, 
we assume that if $\bw'$ is any infinite word with $\maxlgpalpref{\bw'}< \maxlgpalpref{\bw}$ then it is ultimately periodic.
By Lemma~\ref{X1}, 
there exist palindromes $u$ and $v$ with $uv$ primitive such that, 
for all $k \geq 0$, $(uv)^k u$ is a factor of $\bw$. Note that factors $(uv)^k u$ are palindromes. 
Non periodicity of $\bw$ also has the following consequence.

\begin{fct} 
\label{F:R1}
$\pref(\bw) \cap \fact((uv)^\omega)$ is finite.
\end{fct}

\begin{proof}
If $\pref(\bw) \cap \fact((uv)^\omega)$ is infinite, there exists a suffix $s$ of $uv$ such that
$\pref(\bw) \cap \fact(s(uv)^\omega)$ is infinite. Then $\bw = s(uv)^\omega$ and $\bw$ is periodic: a contradiction.
\end{proof}

\noindent
\textbf{Notation}. To continue the proof of Theorem~\ref{T:LGPal}, we need to introduce  notation.
Let $L$ be the length of the greatest prefix of $\bw$ that belongs to $\fact((uv)^\omega)$.
Let recall that $\bw \in \SWHIMPP$.
We let $(\pi_i)_{i \geq 1}$ denote the sequence of length increasing palindromic prefixes of $\bw$.
We now consider first occurrences of words $(uv)^ku$.
Let $k \geq 1$, we let $p_k$ denote the smallest prefix of $\bw$ such that $p_k(uv)^ku$ is also a prefix of $\bw$.
We also let $j_k$ denote the integer such that 
$|\pi_{j_k}| < |p_k(uv)^ku| \leq |\pi_{j_k+1}|$.
In other words, $j_k$ is the greatest integer such that $(uv)^ku$ is not a factor of $\pi_{j_k}$.
Equivalently $j_k+1$ is the smallest integer such that $(uv)^ku$ is  a factor of $\pi_{j_k+1}$.

\begin{fct}
\label{F:R2}
There exist infinitely many integers $i$ such that $2|\pi_i| < |\pi_{i+1}|$.
\end{fct}

\begin{proof}
Consider an integer $i$ such that $i = j_k$ for some integer $j_k \geq 1$.
Assume that $2|\pi_i| \geq |\pi_{i+1}|$. (See Fig.~\ref{fig:1} for an illustration of the hypotheses and the notation of this proof)

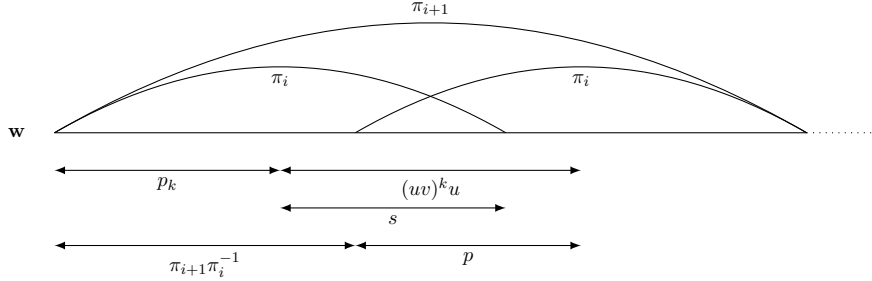
\begin{figure}[h]
\begin{center}
\begin{tikzpicture}[scale=0.5]
\coordinate (G) at (0,0) ;
\coordinate (D) at (20,0);
\coordinate (p) at (12,0);
\coordinate (q) at (8,0);

\node[scale = 0.7] (w) at (-1,0) {$\mathbf{w}$} ;
\draw (G) -- (D) ;
\draw[dotted] (D) -- (22, 0) ;
\draw (0, 0) to[bend left] node[above,	midway,scale = 0.7]{$\pi_{i+1}$}  (20,0)  ;
\draw (G) to[bend left] node[below,	midway,scale = 0.7]{$\pi_{i}$}  (p)  ;
\draw (D) to[bend right] node[below,	midway,scale = 0.7]{$\pi_{i}$}  (q)  ;
\draw[<->,>=latex] (0,-1) -- node[below, midway,scale = 0.7] {$p_k$} (6, -1);
\draw[<->,>=latex] (6,-1) -- node[below, midway,scale = 0.7] { $(uv)^ku$} (14, -1);
\draw[<->,>=latex] (6,-2) -- node[below, midway,scale = 0.7] {$s$} (12, -2);
\draw[<->,>=latex] (0,-3) -- node[below, midway,scale = 0.7] {$\pi_{i+1}\pi_i^{-1}$} (8, -3);
\draw[<->,>=latex] (8,-3) -- node[below, midway,scale = 0.7] {$p$}(14, -3);
\end{tikzpicture}
\caption{Hypotheses and notation in the proof of Fact~\ref{F:R2}}
\label{fig:1}
\end{center}
\end{figure}

By choice of $i$, 
$|p_k(uv)^ku| \geq |\pi_i| \geq |\pi_{i+1}|-|\pi_i|$.
We also have $|p_k| < |\pi_{i+1}|-|\pi_i| \leq |\pi_i|$.
Indeed otherwise the word $(uv)^ku$ would be a factor of the suffix $\pi_i$ of $\pi_{i+1}$.
This contradicts the definition of $i =j_k$.

Consider the suffix $s$ of $\pi_i$ such that $\pi_i = p_k s$.
As $\pi_i$ and $p_k(uv)^ku$ are prefixes of $\pi_{i+1}$.
The word $s$ is a prefix of $(uv)^ku$.
As $(uv)^ku$ and $\pi_i$ are palindromes, $\tilde{s}$ is a prefix of $\pi_i$ (and so of $\bw$)
and a factor of $(uv)^ku$.

Consider now the prefix $p$ of $\pi_i$ of length $|p_k(uv)^ku|-(|\pi_{i+1}|-|\pi_i|)$: 
$p_k (uv)^ku =(\pi_{i+1}\pi_i^{-1})p$ 
As $|p_k| < |\pi_{i+1}|-|\pi_i|= |\pi_{i+1}\pi_i^{-1}|$, 
$p$ is a suffix of $(uv)^ku$.

Observe now that $|p|+|\tilde{s}| \geq |(uv)^ku|$.
Indeed $|p|+|\tilde{s}| = |p|+|s| = 
(|p_k(uv)^ku| - |\pi_{i+1}|+|\pi_i|)+(|\pi_i|-|p_k|) =
|(uv)^ku| - |\pi_{i+1}|+2|\pi_i|$, 
and, $2|\pi_i| \geq |\pi_{i+1}|$.

We have shown that hypothesis $2|\pi_i| \geq |\pi_{i+1}|$ with $i = j_k$ implies that $\bw$
has a prefix ($p$ or $\tilde{s}$) of length at least $|(uv)^ku| /2$ which is a factor of $(uv)^ku$.
Fact~\ref{F:R1} shows that this situation can hold only for a finite number of integers $k$. 
This ends the proof of Fact~\ref{F:R2}.
\end{proof}

\noindent
\textbf{Intermediate step}. Splitting the proof into three cases.

Let us observe now that, in a word $\pi_i$ (as in any finite word), there occur only finitely many factors in the form $(uv)^ku$.
In the following we will always consider integers $k$ such that $(uv)^{k+1}u$ does not occur in $\pi_{j_k}$.
From Fact~\ref{F:R2} and its proof, the following set is infinite
$${\cal I} = \{ (i, k) \mid i = j_k, 2|\pi_i| < |\pi_{i+1}|, (uv)^{k+1}u \text{ is not a factor of } \pi_{i+1} \}.$$

For any $k \geq 0$ with $(j_k, k)$ in ${\cal I}$, exactly one of the following three cases holds:
\begin{description}
\item[Case 1:] $|\pi_{j_k}| \leq |p_k|$,
\item[Case 2:] $|\pi_{j_k}| > |p_k|$ and $|p_k(uv)^ku| \geq |\pi_{j_k+1}|- |\pi_{j_k}|$,
\item[Case 3:] $|\pi_{j_k}| > |p_k|$ and $|p_k(uv)^ku| < |\pi_{j_k+1}|- |\pi_{j_k}|$.
\end{description}

For $n \in \{1, 2, 3\}$, let
${\cal I}_n$ be the set of all $(i, k)$ in ${\cal I}$ such that case $n$ holds.
Sets ${\cal I}_1$, ${\cal I}_2$ and ${\cal I}_3$ form a partition of ${\cal I}$.
So at least one is infinite. 
The proof of Theorem~\ref{T:LGPal} ends with the next three facts, 
each one proving that one of the sets ${\cal I}_1$, ${\cal I}_2$ and ${\cal I}_3$ 
cannot be infinite without contradicting the hypotheses.

\begin{fct}
\label{F:fact1}
Hypothesis``${\cal I}_1$ is infinite" is contradictory
\end{fct}

\begin{figure}[h]
\begin{center}
\begin{tikzpicture}[scale=0.5]
\coordinate (G) at (0,0) ;
\coordinate (D) at (20,0);
\coordinate (q) at (12,0);
\coordinate (p) at (8,0);

\node[scale = 0.7] (w) at (-1,0) {$\mathbf{w}$} ;
\draw (G) -- (D) ;
\draw[dotted] (D) -- (22, 0) ;
\draw (0, 0) to[bend left] node[above,	midway,scale = 0.7]{$\pi_{i+1}$}  (20,0)  ;
\draw (G) to[bend left] node[below,	midway,scale = 0.7]{$\pi_{i}$}  (p)  ;
\draw (D) to[bend right] node[below,	midway,scale = 0.7]{$\pi_{i}$}  (q)  ;
\draw[<->,>=latex] (0,-1) -- node[below, midway,scale = 0.7] {$p_k$} (10, -1);
\draw[<->,>=latex] (10,-1) -- node[below, midway,scale = 0.7] { $(uv)^ku$} (13, -1);
\end{tikzpicture}
\caption{Case 1 with $i = j_k$; elements of ${\cal I}_1$}
\label{fig:2}
\end{center}
\end{figure}
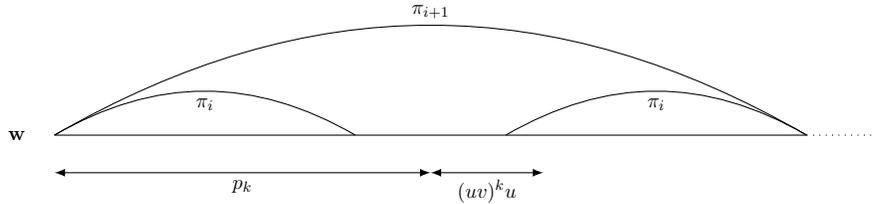

\begin{proof}
Let $(i, k)$ be an element of ${\cal I}_1$ (see Figure~\ref{fig:2}).
We have $|\pi_i| \leq |p_k|$. 
Let $s_k$ be the word such that $\pi_{i+1} = p_k (uv)^ku s_k$.
As $\pi_{i+1}$, $u$ and $v$ are palindromes,
by definition of $p_k$, 
it follows that $|s_k| = |\widetilde{s_k}| \geq |p_k|$.
As $|\pi_i| \leq |p_k|$,
by definition of $p_k$, 
the word $(uv)^ku$ occurs in $\pi_{i+1}$ but does not occur nor overlap the prefix of $\pi_i$.
Thus if ${\cal I}_1$ is infinite, Lemma~\ref{L:main idea} raises a contradiction with $\maxlgpalpref{\bw}$ finite.
\end{proof}

\begin{fct}
\label{F:fact2}
Hypothesis ``${\cal I}_2$ is infinite" is contradictory
\end{fct}

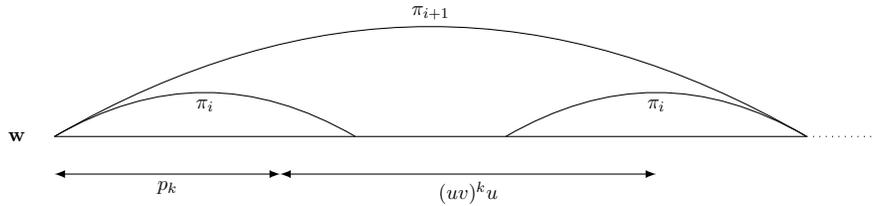
\begin{figure}[h]
\begin{center}
\begin{tikzpicture}[scale=0.5]
\coordinate (G) at (0,0) ;
\coordinate (D) at (20,0);
\coordinate (q) at (12,0);
\coordinate (p) at (8,0);

\node[scale = 0.7] (w) at (-1,0) {$\mathbf{w}$} ;
\draw (G) -- (D) ;
\draw[dotted] (D) -- (22, 0) ;
\draw (0, 0) to[bend left] node[above,	midway,scale = 0.7]{$\pi_{i+1}$}  (20,0)  ;
\draw (G) to[bend left] node[below,	midway,scale = 0.7]{$\pi_{i}$}  (p)  ;
\draw (D) to[bend right] node[below,	midway,scale = 0.7]{$\pi_{i}$}  (q)  ;
\draw[<->,>=latex] (0,-1) -- node[below, midway,scale = 0.7] {$p_k$} (6, -1);
\draw[<->,>=latex] (6,-1) -- node[below, midway,scale = 0.7] { $(uv)^ku$} (16, -1);
\end{tikzpicture}
\caption{Case 2 with $i = j_k$; elements of ${\cal I}_2$}
\label{fig:3}
\end{center}
\end{figure}

\begin{proof}
Let $(i, k)$ be an element of ${\cal I}_2$ (see Figure~\ref{fig:3}).
Let $p$ be the word such that $\pi_i = p_k p$. 
As $|p_k(uv)^ku| \geq |\pi_{i+1}|-|\pi_i|$ and $|\pi_{i+1}| > 2|\pi_i|$,
$|p_k(uv)^ku| > |\pi_i|$, and so,
$p$ is a prefix of $(uv)^k u$.
It follows that $\tilde{p}$ belongs to 
$\pref(\bw)\cap \fact((uv)^ku)$.
By definition of $L$, $|\tilde{p}| \leq L$.
The word $p$, that depends on $k$, can take only a finite number of values.
Possibly replacing ${\cal I}_2$ by an infinite subset, 
we assume that this value is the same for all
elements of ${\cal I}_2$.
Similarly we can assume that for all $(i, k) \in {\cal I}_2$, 
$k \geq 4L+6$ and, as $\lim_{k \to \infty} |\pi_i| = \infty$,
$|\pi_i| > 2 |p|$. 
For $(i, k)$ in ${\cal I}_2$, let $\pi_i'$ denote the word such that $\pi_i = \tilde{p}\pi_i'p$.

Let $s_k$ be the word such that $\pi_{i+1} = p_k(uv)^kus_k$.
As $\pi_{i+1}$, $u$ and $v$ are palindromes, $\pi_{i+1} = \widetilde{s_k} (uv)^ku \widetilde{p_k}$.

\begin{clm}
\label{claim1}
$\widetilde{s_k} = p_k$
\end{clm}

\begin{proof}[Proof of Claim~\ref{claim1}]
By hypotheses of Case~2, 
$|p_k(uv)^ku| \geq |\pi_{i+1}|-|\pi_i| = |p_k(uv)^kus_k| -|\pi_i|$.
Thus $|\pi_i| \geq |s_k|$. 
Words $\widetilde{s_k}(uv)^ku$ and $p_k(uv)^ku$ are prefixes of $\pi_{i+1}$.
Definition of $p_k$ implies that $|s_k| = |\widetilde{s_k}| \geq |p_k|$.

Let $x$ be the word such that $\widetilde{s_k} = p_k x$.
As $|\pi_i| \geq | \widetilde{s_k}|$ and $\pi_i = p_k p$, 
the word $x$ is a prefix of $p$.
Moreover as $\widetilde{s_k}(uv)^ku$ and $p_k(uv)^ku$ are prefixes of $\pi_{i+1}$,
$(uv)^ku$ is a prefix of $x(uv)^ku$.
We have restricted the set ${\cal I}_2$ to elements $(i, k)$ such that $k \geq 4L+6 \geq 1+|p|$.
Thus $k|uv| \geq |uv|+|p| \geq |uv|+|x|$: $xuv$ is a prefix of $(uv)^k$.
Let recall that $uv$ is a primitive word. 
Thus it cannot be an internal factor of $uvuv$. It follows that there exists an integer $\ell$ 
such that $x = (uv)^\ell$.

As $\widetilde{s_k}(uv)^ku$ is a prefix of $\pi_{i+1}$, $(uv)^{\ell+k}u$ is a factor of $\pi_{i+1}$. Just before splitting into three cases the proof, we assume that, for $(i, k) \in {\cal I}$, $(uv)^{k+1}u$ is not a factor of $\pi_{i+1}$.
Hence $\ell = 0$ and $\widetilde{s_k} = p_k$
\end{proof}

Let $\pi$ be the palindrome such that $\pi_{i+1} = \pi_i \pi \pi_i$.
Let recall that
$\pi_{i+1} = p_k(uv)^ku\widetilde{p_k}$, $\pi_i = \tilde{p}\pi_i' p$, $p_k = \tilde{p}\pi_i'$. 
Hence $(uv)^k u = p \pi \tilde{p}$. 
There exist words 
$x$, $y$ and integers $m$, $\ell$ such that $p = (uv)^m x$ with $|x| < |uv|$, 
$uv = xy$ ($\tilde{y}\tilde{x} = vu$ as $u$ and $v$ are palindromes), 
$\pi = y(uv)^\ell u \tilde{y}$.
Observe that $(uv)^k u = p\pi\tilde{p} = (uv)^{\ell+2m+2}u$ and so $k = \ell +2m+2$.
As $k \geq 4L+6$ and $L \geq m$ (as $|p| \leq L$), we have $\ell \geq 2L+4$.

Let $n$ be the greatest integer such that $(\tilde{x}\tilde{y})^n$ is a prefix of $\pi_i$
($n = 0$, if $\tilde{x}\tilde{y}$ is not a prefix of $\pi_i$). 
Let $q$ denote the word $(\tilde{x}\tilde{y})^n$ ($q = \varepsilon$ if $n = 0$ and $q = \tilde{x}(vu)^{n-1}\tilde{y}$ otherwise). 
As $q \in \pref(\bw) \cap \fact((uv)^\omega)$, $|q| \leq L$ and so $n \leq L$.
Thus $\ell +n < k$. The word $p_{\ell+n}$ is a prefix of $\pi_i$.

As $\ell \geq k-2L-2$,
$\lim_{k \to \infty, (i,k) \in {\cal I}_2} \ell = +\infty$, and
as $\bw$ is not ultimately periodic,
$\lim_{k \to \infty, (i,k) \in {\cal I}_2} p_{\ell+n} = +\infty$.
Possibly removing, once again, a finite number of elements of ${\cal I}_2$, 
one can assume $|p_{\ell+n}| \geq |q|+|\tilde{x}\tilde{y}|$.

\begin{clm}
\label{claim2}
$\pi q = y(uv)^{\ell+n}u\tilde{y}$ is the longest palindromic prefix of $\pi p_{\ell+n}$.
\end{clm}

\begin{proof}[Proof of Claim~\ref{claim2}]
Let $\pi'$ be the longest palindromic prefix of $\pi p_{\ell+n}$.
As $|p_{\ell+n}| \geq |q|+|\tilde{x}\tilde{y}|$, the word $\pi'$ begins with the palindrome $\pi''$ defined by $\pi'' =\pi q = y(uv)^{\ell+n}u\tilde{y}$.
As $\pi'$ is a palindrome, $\pi'$ ends with $\pi''$.
There exists a word $z$ such that $\pi' = z y(uv)^{\ell+n}u\tilde{y}$.

By definition of $p_{\ell+n}$, the word $(uv)^{\ell+n}u$ is not a factor of $p_{\ell+n}$.
Thus any occurrence of $(uv)^{\ell+n}u$ in $\pi p_{\ell+n}$ 
must occur in or overlap the prefix $\pi$.
Thus $|zy| \leq |\pi|$ and $\pi' = \pi p$ 
with $p$ both a prefix of $p_{\ell+n}$
and a suffix of $(uv)^{\ell+n}u \tilde{y}$.
As $p$ is a factor of $(uv)^\omega$ and a prefix of $\bw$, $|p| \leq L$.
As $\ell \geq 2L+4$, it follows that the prefix 
$\pi''=y(uv)^{\ell+n}u\tilde{y}$ of $\pi'$
overlaps the suffix $(uv)^{\ell+n}u\tilde{y}$ of $\pi'$ 
by a factor of length at least $2|uv|$.


As $uv$ is primitive, 
$uv$
is not an internal factor of $uvuv$.
Thus $\pi' = \pi(\tilde{x}\tilde{y})^{n'} = y(uv)^{\ell+n'}\tilde{y}$ for some integer $n'$.
As $\pi(\tilde{x}\tilde{y})^{n}$ is a prefix of $\pi p_{\ell+n}$, $n' \geq n$.
Maximality of $n$ in its definition implies $n' = n$, that is, $\pi' = \pi$
\end{proof}

We have already seen that $\lim_{k \to \infty, (i,k) \in {\cal I}_2} \ell = +\infty$.
Consequently $\lim_{k \to \infty, (i,k) \in {\cal I}_2} q^{-1}p_{\ell+n} = +\infty$.
Let $z$ be a prefix of $q^{-1}\bw$. There exist integers $i$, $k$, $\ell$ such that 
$qz$ is a prefix of $p_{\ell+n}$ itself a prefix of $\pi_i$ and for a palindrome $\pi$, $\pi_{i+1} = \pi_i\pi\pi_i$.
Moreover from what precedes, especially from Claim~\ref{claim2}, 
$\lgpal{\pi_i(\pi q)z} = 2 + \lgpal{z}$.
Hence $\maxlgpalpref{q^{-1}\bw} = \maxlgpalpref{\bw}-2 < \maxlgpalpref{\bw}$.
By an initial hypothesis (see the second paragraph of Section~\ref{subsec:proofThLGPAL}), we deduce that $q^{-1}\bw$ is ultimately periodic. So is $\bw$.
Hence Hypothesis ``${\cal I}_2$ is infinite" is contradictory.
\end{proof}

\begin{fct}
\label{F:fact3}
Hypothesis ``${\cal I}_3$ is infinite" is contradictory
\end{fct}

\begin{figure}[h]
\begin{center}
\begin{tikzpicture}[scale=0.5]
\coordinate (G) at (0,0) ;
\coordinate (D) at (20,0);
\coordinate (q) at (12,0);
\coordinate (p) at (8,0);

\node[scale = 0.7] (w) at (-1,0) {$\mathbf{w}$} ;
\draw (G) -- (D) ;
\draw[dotted] (D) -- (22, 0) ;
\draw (0, 0) to[bend left] node[above,	midway,scale = 0.7]{$\pi_{i+1}$}  (20,0)  ;
\draw (G) to[bend left] node[below,	midway,scale = 0.7]{$\pi_{i}$}  (p)  ;
\draw (D) to[bend right] node[below,	midway,scale = 0.7]{$\pi_{i}$}  (q)  ;
\draw[<->,>=latex] (0,-1) -- node[below, midway,scale = 0.7] {$p_k$} (6, -1);
\draw[<->,>=latex] (6,-1) -- node[below, midway,scale = 0.7] { $(uv)^ku$} (10, -1);
\end{tikzpicture}
\caption{Case 3 with $i = j_k$; elements of ${\cal I}_3$}
\label{fig:4}
\end{center}
\end{figure}
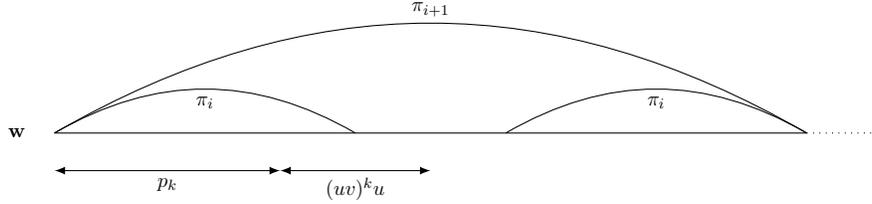

\begin{proof}
Assume by contradiction that ${\cal I}_3$ is infinite.
As in the proof of Fact~\ref{F:fact2}, possibly replacing ${\cal I}_3$ with an infinite subset, as a consequence of Fact~\ref{F:R1}, we can assume $\ell \geq 2+L$.
Let $(i, k)$ be an element of ${\cal I}_3$  (see Figure~\ref{fig:3}).
As $|\pi_i| > |p_k|$, let $p$ denote the suffix of $\pi_i$ such that $\pi_i = p_k p$.
By definition of $p_k$ and $i = j_k$,
$|p_k(uv)^ku| > |\pi_i|$: $p$ is a prefix of $(uv)^ku$.
There exist words $x$, $y$ and integers $m$, $\ell$ such that 
$p = (uv)^mx$, $k = \ell +m+1$, $uv = xy$, $|x| < |uv|$.
Thus $y(uv)^\ell u$ is a prefix of $\pi\pi_i$ where 
$\pi$ is the word such that $\pi_{i+1} = \pi_i\pi\pi_i$ 
(let recall that $2|\pi_i| > |\pi_{i+1}|$).

Let us use the hypothesis $|p(uv)^ku| < |\pi_{i+1}|-|\pi_i| = |\pi_{i+1}\pi_i^{-1}|$.
It implies that $y(uv)^\ell u$ is a prefix of $\pi$.
As $\pi$ is a palindrome (as $\pi_{i+1} = \pi_i \pi \pi_i$ is a palindrome),
$(uv)^\ell u\tilde{y}$  is also a suffix of $\pi$.
Note that $\ell < k$ and so $p_\ell$ is a prefix of $p_k$ so of $\pi_i$.

\begin{clm}\label{claim3}
$\pi$ is the longest palindromic prefix of $\pi p_\ell$.
\end{clm}

\begin{proof}[Proof of Claim~\ref{claim3}]
Let $\pi'$ be the longest palindromic prefix of $\pi p_\ell$.
As $(uv)^\ell u$ does not occur in $p_\ell$, 
any occurrence of $(uv)^\ell u$ must occur in $\pi$ or overlap this prefix of $\pi p_\ell$.
The word $\pi'$ ends with $(uv)^\ell u\tilde{y}$.
So this occurrence must overlap the suffix $(uv)^\ell u\tilde{y}$ of $\pi$.
As in the proof of Claim~\ref{claim2} in the proof of Fact~\ref{F:fact2}, as $\ell \geq 2 +L$,
$\pi' = \pi (\tilde{x}\tilde{y})^n$ for some integer $n$.
If $n\geq 1$, then $y(uv)^{\ell+1}u$
is a prefix of $\pi$ and
$p_k(uv)^{k+1}u$ is a prefix of $\pi_{i+1}$.
This contradicts an earlier hypothesis on $i = j_k$.
Thus $y(uv)^{\ell+1}u$ is not a prefix of $\pi$: we get $n = 0$.
\end{proof}

Let us end the proof of Fact~\ref{F:fact3}.
Observe that $\lim_{k \to \infty, (i,k) \in {\cal I}_3} p_\ell = +\infty$.
Thus for any prefix $z$ of $\bw$, we can find an integer $i$ and a palindrome $\pi$ such that
$\lgpal{\pi_i\pi z} = 2 + \lgpal{z}$. 
It follows that $\maxlgpalpref{\bw} = \maxlgpalpref{\bw}-2$. This is impossible
\end{proof}

\section{Conclusion}

To summarize this paper, observe that the A.~Frid, S.~Puzynina and L.Q.~Zamboni's conjecture could be reformulated as follows. For an infinite word $\bw$ having infinitely many palindromic prefixes, the following assertions are equivalent:
\begin{enumerate}
\itemsep0cm
\item $\bw$ has bounded palindromic lengths of factors;
\item $\bw$ has bounded palindromic lengths of prefixes;
\item $\bw$ has bounded left greedy palindromic lengths of factors;
\item $\bw$ has bounded left greedy palindromic lengths of prefixes;
\item $\bw$ has bounded right greedy palindromic lengths of factors;
\item $\bw$ has bounded right greedy palindromic lengths of prefixes;
\item $\bw$ is periodic;
\item $\bw = (uv)^\omega$ with $u$ and $v$ two palindromes.
\end{enumerate}
Equivalence between assertions 3 to 8 are proved in this paper, and clearly assertion 8 implies assertion 1 which implies assertion 2.
That assertion 1 or assertion 2 implies assertion 8 stays an open problem.

To end this paper, let us mention another related problem.
We first need notation. 
For any finite or infinite word $w$, let $B(w) := \max\{ \pal{p} \mid p$ prefix of $w\}$. 
For any integer $k \geq 1$, let also $B(k) = \min\{ B(\bw) \mid \#{\rm alph}(\bw) = k, \bw \text{ infinite }\}$ where $\#$ denotes the cardinality of a set and ${\rm alph}(\bw)$ is the alphabet of $\bw$. The value $B(k)$ is the least value $B$ for which there exists an infinite word $\bw$ written using $k$ letters and whose palindromic lengths of prefixes are bounded by $B$.

Clearly $B(1) = 1$.
Let us show that $B(2) = 2$ and $B(3) = 3$.
Any infinite binary word begins with a word in the form $aa^ib$ ($i \geq 0$) 
whose palindromic length is $2$: $B(2) \geq 2$. 
As there exist binary words in $\BPLF(2)$, $B(2) = 2$.
Any infinite ternary word begins with a word in the form $ucc^ia$ 
for some different letters $a$, $b$ and $c$ with $i \geq 0$, $|u|_a \geq 1$, $|u|_b \geq 1$ and $|u|_c = 0$. Readers can verify that $\pal{uc} \geq 3$ if $u$ is not a palindrome and $\pal{ucc^ia} \geq 3$ if $u$ is a palindrome. Hence $B(3) \geq 3$. Word $(1213121)^\omega$ belong to $\BPLF(3)$ and so shows that $B(3) = 3$.

Observe also that for any integer $k \geq 1$, $B(k) \leq B(k+1)$. Indeed if $\bw$ is a word written using $k+1$ letters, if $a$ is a letter occurring in $a$, and if $\delta_a(\bw)$ is the word obtained from $\bw$ removing all its occurrences of $a$, then $B(\delta_a(\bw)) \leq B(\bw)$.

Next result shows that $B(2^k) \leq k+1$ and so $B(4) = 3$.

Let $inc$ be the morphism over $\xN^*$ defined by $inc(n) = (n+1)$ for all letters $n$ in $\xN$.
Let $u_1 = 1$, and let, for $n \geq 1$, $u_{n+1} = u_n inc^{2^{n-1}}(u_n) u_n$: $u_2 = 121$, $u_3 = 121343121$, $u_4 = 121343121 565787565 121343121$, \ldots It is straightforward that $u_n$ is written using $2^n$ letters and is of length $3^{n-1}$. Let also $v_n = inc^{2^{n-1}}(u_n)$. Note  that all words $u_n$ and $v_n$ are palindromes. 

\begin{lmm}
For $n \geq 1$, $B(u_{n}) = n$ and $B((u_nv_n)^\omega) = n+1$.
\end{lmm}

\begin{proof}
The proof acts by induction. The result is clearly true for $n = 1$.
Assume $B(u_n) = n$ for some integer $n \geq 1$. It follows that $B(v_n) = n$ also.
Let $p$ be a prefix of $(u_nv_n)^\omega$.
If $p$ is a prefix of $u_n$, $\pal{p} \leq n$ by hypothesis.
If $p  = (u_nv_n)^\ell p_1$ with $p_1$ a prefix of $u_n$ and $\ell \geq 1$, 
let $p_2$ be the word such that $u_n = p_2\widetilde{p_1}$: 
$p = p_2(\widetilde{p_1} v_n (u_nv_n)^{\ell-1} p_1)$. 
As $\widetilde{p_1} v_n (u_nv_n)^{\ell-1} p_1$ is a palindrome, $\pal{p} \leq \pal{p_2}+1 \leq n+1$.
Finally, if $p = (u_nv_n)^\ell u_n p_1$ with $p_1$ a prefix of $v_n$ and $\ell \geq 0$, as $(u_nv_n)^\ell u_n$ is a palindrome,
$\pal{p} \leq \pal{p_1}+1 \leq n+1$. Hence $B((u_nv_n)^\omega) \leq n+1$.

Now let $p$ be a prefix of $v_n$ such that $\pal{p} = n$. As $u_n$ is a palindrome, and as the letters in $u_n$ and in $p$ are different, $\pal{u_np} = 1 + \pal{p}$. Hence $B((u_nv_n)^\omega) = n+1$.

Now as $u_np$ is a prefix of $u_{n+1}$ which itself is a prefix of $(u_nv_n)^\omega$, $B(u_{n+1}) = n+1$.
\end{proof}

Now we are able to state our new problem: for $k \geq 1$ and for $i$ such that $2^{k-1} < i \leq 2^k$, is it true that $B(i) = k+1$? 


\bibliographystyle{plain}
\bibliography{palindromes}

\end{document}